\documentclass[11pt]{article}
\pdfoutput=1
\newif\ifjournal

\journaltrue

\usepackage{nicefrac}
\usepackage[linesnumbered,noend]{algorithm2e}
\usepackage{amsthm,amsmath,amsfonts}
\usepackage[letterpaper, margin=1in]{geometry}

\newtheorem{lemma}{Lemma}[section]
\newtheorem{theorem}[lemma]{Theorem}
        {\hspace*{\fill}$\Box$\par}

\usepackage{epsfig}
\usepackage{amssymb}
\usepackage{amsmath}
\usepackage{amsfonts}

\usepackage{latexsym}
\usepackage{euscript}
\usepackage{amstext}
\usepackage{graphicx}
\usepackage{color}
\usepackage{multicol}
\usepackage{nicefrac}
\usepackage{subfigure}
\usepackage{algorithmicx,algpseudocode}
\usepackage{xspace}
\usepackage{paralist}
\usepackage{balance}
\usepackage{url}

\graphicspath{ {./plots/} }
\DeclareGraphicsExtensions{.pdf}

\usepackage[noend,linesnumbered]{algorithm2e}
\SetKw{KwInput}{Input:}

\newcommand{\StreamAlg}{\textsf{Stream\-All\-Thresholds}\xspace}
\newcommand{\BidirectionalAlg}{\textsf{Bi\-di\-rec\-ti\-on\-al\-Alg}\xspace}
\newcommand{\SmoothAlg}{\textsf{Smooth\-Hi\-sto\-gram\-Alg}\xspace}
\newcommand{\Solution}{\mathsf{Solution}}

\newcommand{\Spread}{\Phi}

\DeclareMathOperator*{\argmax}{arg\,max}

\newtheorem{definition}{Definition}

\begin{document}

\title{Submodular Optimization Over Sliding Windows} 
\author{
 Alessandro Epasto \quad  Silvio Lattanzi\quad Sergei Vassilvitskii\quad Morteza Zadimoghaddam\\
Google\\
\{aepasto, silviol, sergeiv, zadim\}@google.com
}

\maketitle
\begin{abstract}
Maximizing submodular functions under cardinality con\-stra\-ints lies at
the core of numerous data mining and machine learning applications,
including data diversification, data summarization, and coverage problems.
In this work, we study this question in the context of data
streams, where elements arrive one at a time, and we want to
design low-memory and fast update-time algorithms that maintain a good solution. 
Specifically, we focus on the sliding window
model, where we are asked to maintain a solution that considers only
the last $W$ items.

In this context, we provide the first non-trivial algorithm that
maintains a provable approximation of the optimum using space sublinear in the size of the window.
In particular we give a  $\nicefrac{1}{3} - \epsilon$ approximation algorithm that uses space polylogarithmic in
the spread of the values of the elements,
$\Spread$, and linear in the solution size $k$ for any constant $\epsilon > 0$ . At the same
time, processing each element only requires a
polylogarithmic number of evaluations of the function itself. When a
better approximation is desired, we show a different algorithm that, at the cost of using more memory, provides a $\nicefrac{1}{2} - \epsilon$ 
approximation and allows a tunable trade-off between average update time and space. This algorithm matches the best known approximation guarantees for submodular optimization in insertion-only streams, a less general formulation of the problem.

We demonstrate the efficacy of the algorithms on a number of real
world datasets, showing that their practical performance far exceeds
the theoretical bounds. The algorithms preserve high quality solutions in streams with millions of items, while storing a
negligible fraction of them.
\end{abstract}

\section{Introduction}
\label{sec:intro}

Providing concise, timely, and accurate summaries is a critical task facing many modern data driven applications. In myriad scenarios, ranging from increasing diversity~\cite{mirrokniKDD2013} to influence maximization~\cite{KempeKDD2003}, this problem can be viewed as optimizing a submodular function subject to cardinality constraints.  Capturing the property of ``diminishing returns," submodular functions can almost be seen as a silver bullet in data mining and machine learning: they are general enough to model many practical situations, yet allow for simple, and efficient optimization algorithms.  

The classical algorithms for submodular function optimization~\cite{nemhauser1978analysis} were developed for the batch setting. The past decade, however, has seen an increased focus on data streams: situations where the input arrives one element at a time, rather than being presented all at once.  At the cost of sacrificing some accuracy, data streams allow for very fast updates, with the majority of algorithms taking only polylogarithmic time to produce an answer after processing each element. Even as data sizes grow into billions and trillions of items, data stream algorithms remain fast and efficient. 

It is therefore not surprising that submodular function optimization on data streams has received a lot of attention in the past few years~\cite{KDD14,kumar2015fast}. However, previous work has only focused on the insertion-only (or incremental) case where items are only added to, and never removed from the stream. This does not capture the recency constraint, often prevalent in practical applications, where we would like to optimize over the latest data, rather than all of the items seen during the duration of the stream. This is usually captured by considering an optimization over the last $W$ items in the stream in what is known as the {\em sliding window} model, introduced by Datar et al.~\cite{datar2002maintaining}. This model is more general and challenging than the insertion-only case, as the algorithm needs to take into account the items that disappear from the sliding window as time passes.  In this work we study, for the first time, the problem of optimizing submodular functions in the sliding windows model, and develop fast and memory-efficient algorithms with provable approximation guarantees.

\subsection{Applications} 
\label{sect:applications}
Before we proceed, we give two examples of submodular function maximization that have wide applications in practice: maximum coverage and active set selection. We will evaluate our algorithms on these scenarios in Section \ref{sect:exp}. 

{ \bf Maximum coverage.} The maximum coverage problem is a well known NP-Hard problem: given many sets over the same ground set, $U$, select $k$ of them that have the largest union, or jointly ``cover" as many elements as possible. In the sliding window formulation, the sets arrive one at a time, and we can only consider the $W$ most recent arrivals. 

This problem has numerous applications. For instance, the sets might represent content available in an online service (e.g. videos, items to purchase, ads, check-ins in a location-based system). Each set has an associated subset of interested users, our goal is to select $k$ sets to maximize the total number of people interested in at least one item. As relevance of items waxes and wanes, recency is a key factor, and items that first appeared long ago, are no longer considered material\footnote{We note that in practice, more nuanced variations of this basic problem are often used, including assigning weights to users, and allowing partial coverage of users, all of which can be cast as submodular maximization.}.

In other examples, the sets in the input might represent topics (or labels, tags, etc.) covered by a given item, and again we are interested in showing a limited number of items so that we cover as many topics of interest as possible. Other applications of maximum cover in insertion-only streams have been discussed, for instance in~\cite{saha2009maximum}. In our experiments in Section \ref{sect:exp} we show two simple applications of max coverage based on publicly-available data: maintaining a set of recently active points of interest using the Gowalla location-based social network check-ins, and analyzing DBLP co-authorship data to extract a set of recent researchers covering as many fields as possible.

{\bf Active set selection. } 
Another application of submodular maximization lies in the area of data summarization. In this context we want to extract a representative set of $k$ elements from an arbitrary set of items. This setup has many applications in explorative data analysis and visualization, as well as, in speeding up machine learning methods. For instance, in an online system receiving a stream of event updates (e.g. possible security alerts, news stories, etc.) we want to keep track of $k$ informative events to be shown for diagnostic and visualization purposes, or for conducting more in depth analysis.

Here, too, we only want to present recent items from the stream, as older events are less relevant. A concrete instantiation of this problem is that of active set selection with Informative Vector Machines (more details available in \cite{KDD14}), which consists of selecting a set of $k$ items which maximize a submodular function defined on the restricted kernel matrix over the selected items. More precisely let $K_{S,S}$ be the restricted kernel matrix over the items $s_1, \ldots s_{|S|} \in S$ i.e. $K_{S,S}(i,j) = \mathcal{K}(s_i,s_j)$ where $\mathcal{K}(s_i,s_j)$ is the similarity of items $i$ and $j$ according to some symmetric positive definite kernel function $\mathcal{K}$.  In the experimental evaluation, we use the settings of \cite{KDD14}: the items are points in a Euclidean space, $\mathcal{K}=\exp (-\|s_i - s_j\|^2_2/0.75^2)$, and the goal is to find $S$ that maximizes the log-determinant: $f(S) = \frac{1}{2} \log \det (\mathbf{I} + K_{S,S}),$
where $\mathbf{I}$ is the identity matrix of size $|S|$. 

\subsection{Our Contributions} 
In this work we give the first algorithms for monotone submodular function optimization subject to a cardinality constraint over sliding windows, prove bounds on their performance, and empirically demonstrate their effectiveness. Note that algorithms designed for insertion-only streams or off-line settings (e.g. the greedy algorithm) cannot be readily applied in the sliding window case, as items are removed from the window at each update. A na\"ive application of any such method would require at least $\Theta(W)$ time and space to process each new item for a window of size $W$, which is prohibitive. In contrast, we show that sublinear space and time are sufficient: 
\begin{compactitem}
\item In Section \ref{sect:algo-onethird} We give a $\nicefrac{1}{3}-\epsilon$ approximation algorithm that uses memory $O(k \log^2 (k\Spread)/\epsilon^2)$ and needs  only $O(\log^2 (k\Spread)/\epsilon^2)$ calls to the submodular function to process each element. (Here $\Spread$ is the ratio between maximum and minimum values of the submodular function, see Section~\ref{sec:prelim} for details.) The space and time requirements are optimal up to polylogarithmic factors. 
\item We then give an algorithm that achieves a better approximation ($\nicefrac{1}{2}-\epsilon$), at the cost of slower processing, and give a trade-off between update time and total space used by the algorithm (Section \ref{sect:algo-onehalf}). This algorithm matches the approximation guarantees of the best known insertion-only algorithm~\cite{KDD14}.
\item We describe practical considerations used to further improve the runtime of the algorithm in Section \ref{sect:implementation}.
\item In Section \ref{sect:results} we evaluate our algorithms on real world datasets, and empirically demonstrate their accuracy and scalability. 
\end{compactitem} 
Finally, we note that one challenging open problem in the sublinear algorithm literature is to
understand the relationship between different streaming models (see the list of Open Problems in Sublinear Algorithms~\cite{sublinear}). In this context, our results are a significant contribution toward the solution of the problem for submodular functions.
\section{Related work}
There are two lines of work that are related to our paper: literature on submodular optimization and sliding windows. We briefly describe the most relevant results in each area.

{\bf Submodular optimization.} The past decade has seen significant growth in applications of submodular optimization in multiple data mining and machine learning scenarios. The diminishing returns property captures the properties necessary to model the challenging task of selecting representatives among massive amounts of data. These representatives are used as seeds in influence maximization~\cite{KempeKDD2003} and information diffusion networks~\cite{bakshy2012role}, cluster centers in exemplar based clustering~\cite{FreyNIPS2005}, informative vectors in active set selection~\cite{KrauseNIPS2013}, diverse sets in coverage problems~\cite{mirrokniKDD2013}, and in document summarization~\cite{BilmesACL2011}.

The classic solution for submodular optimization with cardinality constraints is the well-known greedy
algorithm introduced by Nemhauser et al.~\cite{nemhauser1978analysis}. Recent years have seen a lot of attention paid to designing faster algorithms for influence maximization~\cite{KDD14,badanidiyuru2014fast,leskovec2007cost,citeulike:10637610}. The most relevant work for us is~\cite{KDD14} where Badanidiyuru 
et al.~introduce the first efficient streaming algorithm for submodular optimization. Specifically,  for the problem of monotone submodular function optimization subject to a cardinality constraint Badanidiyuru 
et al.~\cite{KDD14} give a $1/2 -\epsilon$ approximation while using memory $O(k\log(k)\epsilon^{-1})$ for any $\epsilon>0$. 
In this paper we build on the techniques introduced in that paper to design our algorithms.

{\bf Streaming on sliding windows.} The sliding window model has been introduced by Datar, 
Gionis, Indyk and Mowani in~\cite{datar2002maintaining}. After its introduction the model received a
lot of attention~\cite{DBLP:conf/soda/BravermanLLM16,Braverman07,DBLP:journals/algorithmica/ChanLLT12,gibbons2002distributed,DBLP:journals/algorithms/TingLCL11}. An important concept in this area of research is the concept of Smooth Histograms introduced 
in~\cite{Braverman07} by Braverman and Ostrovsky. Our $\nicefrac13 - \epsilon$ 
approximation algorithm can be seen as an extension of the Smooth Histograms for Submodular functions.  To the best of our knowledge no previous work has addressed the problem of submodular maximization in the sliding window setting with approximation guarantees. 
\section{Preliminaries}
\label{sec:prelim}
Let $V$ be a ground set of elements. A function $f : 2^V \rightarrow \mathbb{R}^{\geq 0}$ is said to be submodular, 
if for all sets $S \subset T \subset V$ and all elements $v \not\in T$, 
$$f(S \cup \{v\}) - f(S) \geq f(T \cup \{v\}) - f(T).$$
In other words, the additional benefit of element $v$ is no larger when added to $T \supseteq S$. To simplify notation, for an element $v \in V$, and set $S \subset V$, let 
$$f'_S(v) = f(S \cup \{v\}) - f(S),$$
denote the incremental value of adding element $v$ to set $S$. 

A submodular function $f$ is monotone, if for any $S \subseteq T$, $f(T) \geq f(S)$.  In this work we focus on optimizing monotone submodular functions, subject to a cardinality constraint. For $k \in \mathbb{Z}$, let 
$$f_k(V) =  \max_{S \subseteq V : |S| = k} f(S).$$

It is well known~\cite{nemhauser1978analysis} that the simple greedy algorithm that starts with $S = \emptyset$, and repeatedly adds the element $v$  that maximizes $f'_S(v)$  achieves a $(1 - \nicefrac{1}{e})$ approximation to the optimum solution. Moreover this approximation ratio is the best possible, unless $\mathcal{P} = \mathcal{NP}$. 

{\bf Streaming Algorithms.} 
Data streams are a common way to design algorithms for very large datasets, see ~\cite{aggarwal2007data,mcgregor2014graph} for a survey. In this setting, elements arrive one at a time, and the goal of the algorithm designer is to maintain a (nearly) optimal solution. A trivial approach is to store all of the elements, and recompute the solution from scratch every time. Such an approach is obviously inefficient, it requires both large memory (to store all of the elements), and large update time upon reading every element.  In evaluating streaming algorithms, we will focus on these two metrics.  
For a stream of length $n$, the goal is to find algorithms that require sublinear memory, and update time, with the gold standard having both be $O(\text{polylog} (n))$.  

In this work, we are specifically interested in the sliding window model over data streams. Consider a stream $v_1, v_2, \ldots$.
Without loss of generality, we assume no item of zero value is present, i.e. $f(\{ v_i\}) > 0, \forall i$. Notice that such items can be discarded without affecting the objective function value
because such $v_i$ have zero incremental value to every set. 
Let $\Delta = \max_{v \in V} f(\{v\})$ be the maximum value of a set containing a single element in $V$. 
We also let $$\Spread = \frac{\max_{v \in V} f(\{v\})}{\min_{v \in V} f(\{v\})},$$ be the ratio of maximum to minimum singleton values. Our algorithms do not need to know $\Spread$ (it only appears in space and computation upper bounds). Although we present the algorithms as they need to know $\Delta$, in Subsection~\ref{sect:implementation} we show how to relax this assumption without loss of generality. 

Let $W \in \mathbb{Z}$ be the size of the sliding window. At each point in time $t \geq W$  the {\em active window}, $A_t$, is the set that contains the last $W$ elements in the stream: $A_t = \{v_{t - W+1}, v_{t - W + 2}, \ldots, v_t\}$. For $t < W$, we let $A_t = \{v_1, v_2, $ $\ldots, v_t\}$. We are interested in computing sets $S_1, S_2, \ldots$ of cardinality $k$ such that at every time $t$, $f(S_t)$ is within a small constant factor of $f_k(A_t)$. 

Similarly to streaming algorithms, an obvious approach is to store the whole window $A_t$, and recompute the optimal function on $A_t$ at every time step. In this work we will show how to compute an approximately optimal solution to $f$ using much less space, and with much faster update time. 

\section{A $(\nicefrac{1}{3}-\epsilon)$-approximation algorithm}
\label{sect:algo-onethird}

In this section we present an algorithm that uses polylogarithmic memory and 
update time  to compute a $(\nicefrac{1}{3}-\epsilon)$-approximation for the submodular 
maximization problem with cardinality constraints.

A key ingredient in our analysis is the concept of Smooth Histograms introduced by Braverman
and Ostrovsky in~\cite{Braverman07}. Before presenting our solution, we briefly review
the main ideas presented in~\cite{Braverman07}.

\textbf{Smooth Histograms.}
The key idea behind smooth histograms is to identify and maintain a subset of indices $x_1$, $x_2$, $\ldots$, $x_s$, such that we only consider the intervals starting at $x_i$ and ending at $t$. If we can prove that one of these intervals leads to an approximately optimal solution, then we can proceed by running $s$ copies of a streaming approximation algorithm in parallel, one starting at each index $x_i$.  

The main challenge is in identifying the right set of indices. It is easy to show that simple ideas---for example evenly partitioning the window into $W/s$ equally spaced starting points, or using reservoir sampling to maintain $s$ random starting points---do not work, in particular because the partitioning must depend on the value of the objective function on the different sub-intervals.  

Braverman and Ostrovsky show that for a well behaved function, $g$, it is possible to maintain such a set of indices. The high level idea is to look at the function values, and insist that for any three successive indices, $x_{\ell - 1}, x_\ell, x_{\ell + 1}$ the value of $g(x_{\ell+1}, t) \leq (1 - \beta) g(x_{\ell-1}, t)$ for some constant $\beta$. Here $g(a,b)$ is the value of function $g$ on the interval $[a,b]$ of elements, i.e. $\{v_a, v_{a+1}, \ldots, v_b\}$. In this case the total number of indices is bounded by $O(\log_{1+\beta} H)$, where $H$ is the ratio between the maximum and minimum values of $g$.  However, the approximation guarantees only hold for a certain subset of functions. More precisely, 
\begin{definition}
A function $g$ is $(\alpha, \beta)$-smooth if for all indices $a < b < c < d$  we have that:
$$(1-\beta)g(a, c) \leq g(b, c) \implies (1-\alpha)g(a, d) \leq g(b, d).$$
\end{definition}
Braverman and Ovstrovsky then show how to maintain polylogarithmically many indices to get a $1-\alpha$ approximation to an $(\alpha, \beta)$ smooth function $g$. They further extend their results to the setting when $g$ cannot be computed exactly in a streaming setting, but can only be approximated to a factor of $(1 - \epsilon)$. They adapt the analysis (Theorems 2 and 3 in ~\cite{Braverman07}) to show that this results in a $1 - 5\epsilon$ approximation. 

Thus following their analysis, the resulting algorithm gives non-trivial results only when $\epsilon < \nicefrac{1}{5}$.
In our problem, we are interested in computing $g=f_k$, and there exists no $1-\epsilon$ approximation to estimate it. 
 For the submodular maximization problem 
with cardinality constraints, the best streaming algorithm achieves a $\nicefrac12$ approximation.
Furthermore, unless $\mathcal{P}=\mathcal{NP}$, there does not exist any better than 
$1-\nicefrac{1}{e}$ approximation algorithm for submodular maximization with cardinality
constraint~\cite{Feige98}. Thus, we cannot apply their techniques directly in our case. 

Nonetheless, in the rest of this section we show how one can use properties of
submodular functions to adapt the smooth histogram framework and obtain an efficient $(1/3 - \epsilon)$-approximation algorithm. 

\subsection{An insertion only algorithm}
\label{sect:algo-insertionly}

Our first building block is a streaming algorithm that can approximate $f_k$ efficiently. 
We present Algorithm~\ref{alg:StreamingAllThresholds}
(named \StreamAlg), which  is an extended version of ThresholdStream algorithm 
introduced in~\cite{kumar2015fast} and that uses similar techniques to the one
developed in~\cite{KDD14}.
Algorithm~\ref{alg:StreamingAllThresholds} takes a stream of elements $u_1, u_2, \dots$ (in our algorithm this stream is often a sub-stream of the original stream $v_1, v_2, \dots)$. 
Given a value of $\delta > 0$ which we will fix later, we consider a set of $m = \lfloor \log_{1+\delta} 2k \Delta/f(u_1) \rfloor$ thresholds, 
$$T = \left\{\frac{f(u_1)}{2k}, \frac{(1+\delta)f(u_1)}{2k}, \frac{(1 + \delta)^2f(u_1)}{2k}, \ldots, \frac{(1+\delta)^mf(u_1)}{2k}\right\}.$$ 

For each threshold $\tau \in T$, we maintain a feasible solution $S_\tau$ which is initialized with the empty set. At time $t$, when $u_t$ arrives, we add it to the solution if $\left|S_{\tau}\right| < k$ and $f'_{S_{\tau}}(u_t) \geq \tau$. At any time $t$ the current solution is the best among the candidate solutions $\{S\}_\tau$,  i.e. $\Solution_t = \max_{\tau} f(S_{\tau})$. The pseudocode is shown in Algorithm~\ref{alg:StreamingAllThresholds}. 

\begin{algorithm}[t!]
\label{alg:StreamingAllThresholds}
\caption{\StreamAlg}
\KwInput{Stream of elements $u_1, u_2, \cdots$, and $\delta$}\;
Let $m = \lfloor \log_{1+\delta} 2k\Delta/f(u_1) \rfloor$.\\
Let $T  = \{\frac{f(u_1)}{2k}, \frac{(1+\delta)f(u_1)}{2k}, \frac{(1 + \delta)^2f(u_1)}{2k}, \ldots, \frac{(1+\delta)^mf(u_1)}{2k}\}.$\\
\For{$t=1, 2, \cdots$}{
	\ForAll{$\tau \in T$}{
		\If{$f'_{S_{\tau}}(u_t) \geq \tau \wedge |S_{\tau}| < k$}{
			Add $u_t$ to $S_{\tau}$
		}
	}
	$\Solution_t \gets \max_{\tau} f(S_{\tau})$\;
}
\end{algorithm}
We now give a lower bound on the performance of Algorithm ~\ref{alg:StreamingAllThresholds}. For the analysis, let $h(A)$ be the output of 
Algorithm~\ref{alg:StreamingAllThresholds} on the stream $A$ of elements. 

\begin{lemma}\label{lem:ThresholdGuarantee}
For any non-empty set $B \subseteq A$ with $|B| = k' \leq k$, we have 
$h(A) \geq (1-\delta)\frac{k}{k+|B|}f(B)$. Equivalently, for any $1 \leq k' \leq k$, 
$h(A) \geq (1-\delta)\frac{k}{k+k'}f_{k'}(A).$
\end{lemma}

\begin{proof}
We first note that since $f_{k'}(A)$ is at least $f(B)$ by definition of $f_{k'}$, we only need to prove $h(A) \geq (1-\delta)\frac{k}{k+k'}f_{k'}(A)$.
By definition of $\Delta$ and submodularity of $f$, we have that 
$f_{k'}(A) \leq k' \Delta$, and therefore $f_{k'}(A)/(k+k')$ is at most $\Delta/2$.  
One the other hand, we know $f_{k'}$ is at least $f(u_1)$,  consequently,  $f_{k'}(A)/(k+k')$ is at least $f(u_1) /2k $. 
Therefore there exists some $(1-\delta)f_{k'}(A)/(k+k') \leq \tau \leq f_{k'}(A)/(k+k')$  in set  $T$. 
 We proceed to prove the claim for $S_\tau$ which lower bounds the value of $h(A)$. 

There are two cases. If the size of $|S_{\tau}|$ is $k$, then:
$$h(A) \geq  f(S_{\tau}) \geq k \tau \geq (1-\delta)k f_{k'}(A)/(k+k').$$ 

Otherwise, consider an element $u_t \in A \setminus S_{\tau}$, which was not selected. Then
$f'_{S^{t-1}_{\tau}}(u_t) < \tau$ where $S^{t-1}_{\tau}$ is the subset of elements of $S_{\tau}$ that arrive before time $t$. By submodularity, we also have $f'_{S_{\tau}}(u_t) < \tau$. We conclude by:
$$f_{k'}(A) - f(S_\tau) \leq \sum_{x \in f_{k'}(A) \setminus S_{\tau}} f'_{S_{\tau}}(x) \leq \frac{k'}{k+k'}f_{k'}(A),$$
where the first inequality follows from the property of submodular functions, see for example Lemma 5 of \cite{bateni2010submodular}. Therefore,  $f(S_{\tau}) \geq \left(1-\nicefrac{k'}{(k+k')}\right)f_{k'}(A)$, 
which proves the claim. 
\end{proof}

\subsection{The sliding window algorithm}
\label{sect:algo-smooth}

Now we are ready to formulate our $(\nicefrac13-\epsilon)$-approximation algorithm. To
solve our problem we introduce the concept of Submodular Smooth Histograms
inspired by the Smooth Histograms in~\cite{Braverman07}.

A Submodular Smooth Histogram consists of $s$ indices $x_1, x_2, \cdots, x_s$
where the last index $x_s$ is equal to the current time, $t$ and represents the end of
the sliding window. At initialization, $t = 1$, and we set $s=1$, $x_1=1$. 
 
During the algorithm we run $s$ instances of our streaming algorithm concurrently. 
Algorithm $\StreamAlg_i$ is responsible for processing the stream that 
starts with $x_i$ and processes all elements after that unless we decide 
to terminate the algorithm.  At time $t$, 
when an element $v_t$ arrives, it is processed by all $s$ instances 
of \StreamAlg.

Furthermore we also initiate a new instance of \StreamAlg that 
is responsible for the stream that starts with $v_t$. 
Formally, we increment $s$ and set the new $x_s=t$. 

We now show how to update the indices
$x_1, x_2, \cdots, x_s$ to keep $s$ bounded while keeping a good
approximation.  
Recall that $h(A)$ 
is the output of \StreamAlg on window $A$. Abusing notation slightly, we also let $h(a,b)$ be the 
value of function $h$ on the window that starts with index $a$ and ends with index $b$. 
We have two main operations to maintain the indices.
First, if index $x_{i+1}$ has expired: i.e. $x_{i+1} < t - W + 1$, then we remove index $x_i$ for any $0 < i < s$. Second, whenever we have $h(x_{i+2},t) \geq (1-\beta)h(x_i,t)$ for some $0 < i < s$ we remove index $x_{i+1}$. Any time an index is removed the corresponding algorithm is terminated. At any point in time $t$ the current solution $ \Solution_t = h(x_1,t)$ if $x_1$ is not expired and $h(x_2,t)$ otherwise.
In Algorithm~\ref{alg:submodularSmooth} we give the pseudocode
that maintains Submodular Smooth Histograms. 

\begin{algorithm}[t!]
\label{alg:submodularSmooth}
\caption{Submodular Smooth Histograms Algorithm}
\KwInput{A stream of elements $v_1, v_2, \ldots$, parameters $\beta, \delta$, Window size $W$}\;
Initialize $s \gets 0$\;
\ForAll{$t \in \{1, 2, \ldots\}$}{
	$s \gets s+1$\;
	$x_s \gets t$\;
	Initiate a new instance of Algorithm~\ref{alg:StreamingAllThresholds} that processes the stream starting from $x_s$\;
	// Keep at most one expired index.\\
	\ForAll{ $0 < i < s$}{
		\If {$x_{i+1} < t - W + 1$}{
			Remove $x_i$, terminate Algorithm~\ref{alg:StreamingAllThresholds} associated with $x_i$, and shift other indexes accordingly\;
			$s \gets s - 1$\;
		}
	}
	Pass $v_t$ to all $s$ running instances of Algorithm~\ref{alg:StreamingAllThresholds}\;
	// Delete indices that are no longer useful.\\
	\While{$\exists 0 < i < s: h(x_{i+2},t) \geq (1-\beta)h(x_i,t)$}{
		Remove $x_{i+1}$, terminate Algorithm~\ref{alg:StreamingAllThresholds} associated with $x_{i+1}$, and shift the remaining indexes accordingly\;
		$s \gets s - 1$\;
	}
	\If{$x_1 = \max(1, t - W + 1) $}{
		$\Solution_t \gets h(x_1,t)$
	} \Else{$\Solution_t  \gets h(x_2,t)$}
} 
\end{algorithm}

We first show the main property of Submodular Smooth Histograms which is maintained by 
Algorithm~\ref{alg:submodularSmooth}. 

\begin{lemma}\label{lem:SubmodularSmoothProperty}
For any time $t$ and $1 \leq i < s$, we either have $x_{i+1} = x_i +1$ or there exists 
some $t' \leq t$ such that $h(x_{i+1}, t') \geq (1-\beta)h(x_i, t')$.
\end{lemma}

\begin{proof}
Let $t'$ be the first time $x_{i+1}$ becomes the successor of $x_i$ in the smooth 
histogram. If this event occurred due to the removal of some  $x'$ that was 
between $x_i$ and $x_{i+1}$,  the condition of the while loop 
 ensures that 
$h(x_{i+1},t') \geq (1-\beta)h(x_i,t')$. Otherwise, $x_{i+1}$ became the successor 
of $x_i$ when $x_{i+1}$ was added to the smooth histogram. But we never remove
 the last index of the histogram, so the last index was equal to the previous end of
  sliding window $x_{i+1}-1$, 
  therefore $x_i = x_{i+1}-1$.
\end{proof}

We are ready to show that with a judicious choice of $\delta$ and $\beta$, 
Algorithm~\ref{alg:submodularSmooth} is a $(\nicefrac13-\epsilon)$-approximation algorithm. 

\begin{theorem}
For any $\epsilon > 0$, Algorithm~\ref{alg:submodularSmooth} with 
$\beta = \delta = \epsilon/2$ is a $(\nicefrac13-\epsilon)$-approximation for 
submodular maximization with a cardinality constraint in sliding window model. 
\end{theorem}

\begin{proof}
Let $x_1$ and $x_2$ (if it exists) be the first two indices of the smooth 
histogram right after the update operations are done for a newly 
arrived element $v_t$. 
Note that the start of the active window $A_t$ is in the range $[x_1, x_2]$. 
Lemma~\ref{lem:SubmodularSmoothProperty} implies that either 
$x_2=x_1+1$ or $h(x_2,t') \geq (1-\beta)h(x_1, t')$ at some $t' \leq t$. 
If $x_2 = x_1+1$, the start of $A_t$ is be equal to either $x_1$ 
or $x_2$. In this case, we have calculated $h(A_t)$ the result of 
\StreamAlg on window $A_t$ and Algorithm~\ref{alg:submodularSmooth} 
will return $h(A_t)$ as the result. Using Lemma~\ref{lem:ThresholdGuarantee},
we have $h(A_t) \geq (1-\delta)\frac{f_k(A_t)}{2}$ which proves the claim. We note that if $x_2$ does not exist, the start of the active window is $x_1$ and the claim is proved in a similar manner. 

In the other case, we have $h(x_2,t') \geq (1-\beta)h(x_1, t')$ for some
$t' \leq t$. Now if $h$ was $(\alpha, \beta)$-smooth we would be done; in the remaining part 
of the proof we show how to use submodularity instead of smoothness
to prove the claim.

Let $OPT$ be the optimal solution on the interval $(x_1, t)$, formally: 
$$OPT = \argmax_{S \subseteq \{v_{x_1}, v_{x_1+1}\ldots, v_{t}\} \wedge |S| \leq k} f(S).$$ By definition of $f_k$,
$f(OPT) \geq f_k(A_t)$. 

We begin by splitting $OPT$ into two sets, those elements appearing before and after $t'$. 
Let $OPT_1 = $ $OPT \cap \{v_{x_1}, v_{x_1+1} $ $\ldots, v_{t'}\}$ and $OPT_2 = OPT \cap \{v_{t'+1}, \ldots, v_t\}$. Let $k_1 = |OPT_1|$ and $k_2  = |OPT_2|$.  Similarly, let $f_1 = f(OPT_1)$ and $f_2 = f(OPT_2)$. 
By submodularity,
\begin{equation}
\label{eqn:sub}
f(OPT) \leq f(OPT_1) + f(OPT_2).
\end{equation}

Moreover, Lemma~\ref{lem:ThresholdGuarantee} implies that 
\begin{equation}
\label{eqn:split}
h(x_1, t') \geq (1 - \delta)\frac{kf_1}{k + k_1} \text{ and } h(x_2, t) \geq (1 - \delta)\frac{kf_2}{k + k_2}.
\end{equation}
By monotonicity of Algorithm ~\ref{alg:StreamingAllThresholds}, we have: $h(x_2, t) \geq h(x_2, t') \geq (1 - \beta) h(x_1, t')$.
We can now bound $h(x_2, t)$
 \begin{align*} 
 &\geq (1-\beta)(1-\delta)\max\left(\frac{k}{k+k_1}f_1, \frac{k}{k+k_2}f_2\right)\\
 &\geq k(1 - \epsilon) \max \left(\frac{f(OPT) - f_2}{k + k_1}, \frac{f_2}{k + k_2} \right),
 \end{align*}
 where the first inequality follows by Equation~\ref{eqn:split}, the fact that $h(x_2, t) \geq (1 - \beta) h(x_1, t')$,
 and the setting of $\beta$ and $\delta$, and the second from Equation~\ref{eqn:sub}.
 
 For ease of notation, let $\mu = f_2/f(OPT)$. Clearly $\mu \in [0, 1]$. It is possible to verify that 
 \begin{equation}
 \label{eqn:mu}
 \max\left(\frac{1 - \mu}{2k - k_2}, \frac{\mu}{k+k_2}\right) \geq \frac{1}{3k},
 \end{equation} 
 as the maximum is achieved at $k_2 = 3k\mu - k$.  Continuing to bound $h(x_2, t)$:
\begin{align*}
&\geq k (1 - \epsilon) f(OPT) \max \left(\frac{1 - \mu}{2k - k_2}, \frac{\mu}{k + k_2}\right) \\
&\geq k (1 - \epsilon) f(OPT) \frac{1}{3k} \geq \frac{1}{3} (1 - \epsilon) f(OPT)  \geq \frac{1}{3} (1 - \epsilon) f(A_t), 
\end{align*}
where the first inequality follows from definition of $\mu$, and the second from Equation \ref{eqn:mu}; which concludes the proof. 
 \end{proof}

We now state a bound on the memory and the update time of Algorithm ~\ref{alg:submodularSmooth}. 
\begin{theorem}
Algorithm ~\ref{alg:submodularSmooth} with 
$\beta = \delta = \epsilon/2$ has an update time of $O(L\log^2(k\Spread)/\epsilon^2)$ per element and uses memory $O(k\log^2(k\Spread)/\epsilon^2)$ where $L$ is an upper bound on the time for each evaluation of function $f$. 
\end{theorem}
\ifjournal
\begin{proof}
The update operations of the while loop in Algorithm~\ref{alg:submodularSmooth} make sure that $h(x_{i+2}) < (1-\beta) h(x_i)$ for any $0 < i < s$. By definition of $h$ and submodularity of $f$, $h$ cannot be larger than $k\Delta$, and is lower bounded by $\min_{v \in V} f(\{v\})$. Therefore $s$ is at most $O(\log_{1+\beta} (k\Spread)) = O(\log(k \Spread)/\beta)$. Algorithm~\ref{alg:submodularSmooth} keeps running $s$ instances of Algorithm~\ref{alg:StreamingAllThresholds}. Each of these instances maintains $|T| = \log_{1+\delta} (k\Delta/f(u_1))$ $= O(\log(k\Spread)/\delta)$ of sets of size at most $k$. Every new element is considered for addition to each of these $|T|$ sets in each of the $s$ instances. 
We also note that the update operations of while loop can be done in time $O(s^2)$ because there can be at most $s$ removals, and each takes $O(s)$ time to find. 
Therefore the update time per element is $O(Ls|T| + s^2)$ $= O(L\log^2(k\Spread)/(\epsilon^2))$ and total memory is  $O(k\log^2(k\Spread)/(\epsilon^2))$.
\end{proof}
\fi

\section{A $(\nicefrac12 - \epsilon)$-Approximation Algorithm}
\label{sect:algo-onehalf}

In this section we present a $1/2 - \epsilon$ approximation algorithm that uses more memory and 
amortized update time to get a better approximation.

The algorithm is based on two main ideas. The first one is to split the entire stream in sub-windows of size 
$W' \leq W$ and to run a variant of the \StreamAlg starting from the first element of each sub-window. 
Each sub-window $i$ consists of elements that arrive at times $(i-1)W'+1, (i-1)W'+2, \dots, iW'$. This 
guarantees that when the first element of the sliding window is aligned with the start of a sub-window we 
can obtain a $1/2 - \epsilon$ approximation just by using the streaming algorithm started at the sub-window.

Unfortunately the situation is more complex when the first element of the sliding window lies inside a 
sub-window. In fact, there is no stream that would work natively. The second idea behind our algorithm is to 
run a variant of \StreamAlg algorithm first backward from the end of each sub-window to the beginning 
of each sub-window and then forward from the beginning of each sub-window and onwards. In particular, for every 
sub-window $i$, every threshold $\tau$ and every $(i-1)W'+1 \leq t' \leq iW'$, we build the sets $S_\tau^{i,t'}$ such 
that elements are added in $S_\tau^{i,t'}$ by analyzing sequentially elements in $iW', iW' -1, \dots, t', iW'+1,  
iW'+2, \dots$ and by adding an element if and only if the marginal contribution of the element to the value of 
the set is at least $\tau$ and if the set is smaller than $k$. Now there are two key observations to make. 
First, if the first element of the sliding windows arrives at time $t'$ we can use the sets $S_\tau^{i,t'}$ for 
different values of $\tau$ to solve the problem. Second, when we consider the family of sets
$\mathcal{S}_\tau^i=\cup_{t'} S_\tau^{i,t'}$, the family contains at most $k+1$ distinct sets (because going
backwards we add at most $k$ elements) so we can store only those sets and use them to solve the 
problem. In the remainder of this section we formalize this reasoning to get a
$1/2 - \epsilon$ approximation algorithm.

We start by introducing some additional notation. For every sub-window $i$, we define a set of thresholds 
$T_i =$ $\Big\{\frac{f(v_{iW'})}{2k}$, $\frac{(1+\delta)f(v_{iW'})}{2k}$, $\frac{(1 + \delta)^2f(v_{iW'})}{2k}$, $\ldots,
\frac{(1+\delta)^{m_i}f(v_{iW'})}{2k}\Big\}$ where $m_i$ is $\lfloor \log_{1+\delta} 2k\Delta/f(v_{iW'}) \rfloor$. For every $\tau 
\in T_i$, we first compute a single backward pass from the last element in sub-window $v_{iW'}$ and 
end by the first element of the sub-window $v_{(i-1)W'+1}$. In this pass, we add any item with 
marginal value at least $\tau$ to set $B_{\tau}^{i}$ as long as  $|B_{\tau}^{i}|$ remains at most $k$.

By definition $B_{\tau}^{i}$ contains at most $k$ elements; let $j_1 > j_2 > \dots > j_k$ be the indices of the elements $v_{j_1}, v_{j_2}, \dots, v_{j_k} \in B_{\tau}^{i}$. We define $S_\tau^{i,t'}= \cup_{j_{\ell} \geq t'}v_{j_{\ell}}$ 
as the set of elements in $B_{\tau}^{i}$ inserted at or after time $t'$. In our algorithm we do not keep all
$S_\tau^{i,t'}$, but we restrict our attention only to the set $S_\tau^{i,t'}$ for $t' \in \{j_0,j_1, j_2, \dots, j_k\}$ where $j_0 = iW'$.
We also define $\mathcal{S}_\tau^i = \cup_{t' \in \{j_0, j_1, j_2, \dots, j_k\}} S_\tau^{i,t'}$. 
We note that  $|B_{\tau}^{i}| < k$, so there will be at most $k+1$ sets in $\mathcal{S}_\tau^i$.
Finally, to handle the initial elements in the stream, we define  $T_0  =$ $\Big\{\frac{f(v_1)}{2k}$, $\frac{(1+\delta)f(v_1)}{2k}$, 
$\frac{(1 + \delta)^2f(v_1)}{2k}$, $\ldots,$ $\frac{(1+\delta)^{m_0}f(v_1)}{2k}\Big\}$ where $m_0$ is 
$\lfloor \log_{1+\delta} 2k\Delta/f(v_1) \rfloor$. We also initialize set $S_{\tau}^{0,0} = \emptyset$ for 
any $\tau \in T_0$.

Our algorithm has two steps. At first, if needed, it runs the backward algorithm to compute
$S_{\tau}^{i,t'}$. Then, it adds the last element in the stream, $v_t$, to all $S_{\tau}^{i,t'} \in \mathcal{S}_{\tau}^i$,
for every $i_t \leq i \leq \lceil t/W' \rceil -1$ (all active sub-windows) and $\tau \in T_i$, if its marginal impact is large enough. Here we let $i_t = \max\{0, \lceil (t-W+1)/W' \rceil\}$ be the first active sub-window. Finally, we 
set the solution of active window $A_t$ to be $\max_{\tau \in T_{i_t}} S_{\tau}^{i_t, t^*}$ where  
$S_{\tau}^{i_t, t^*}$ is the set in $\mathcal{S}_{\tau}^{i_t}$ with minimum $t^*$ such that $S_{\tau}^{i_t, t^*} \subseteq A_t$. 
We call the algorithm \BidirectionalAlg, and show the pseudo-code in Algorithm~\ref{alg:BidirectionalAlgorithm}. 

\begin{algorithm}[t!]
\label{alg:BidirectionalAlgorithm}
\caption{\BidirectionalAlg}
\KwInput{Stream of elements $v_1, v_2, \cdots$, sub-window size $W' \leq W$ and $\delta$}\;
$S_{\tau}^{0,0} \gets \emptyset$ for each $\tau \in T_0$\;
\For{$t = 1, 2, \cdots$}{
	// Initialize $S_{\tau}^{i, t'}$\;
	\If{$t = iW'$ for some integer $i$}{
		\For{$\tau \in T_i$}{
			$B_{\tau}^{i} \gets \emptyset$\;
			$S_{\tau}^{i, t} \gets \emptyset$\;
			\For{$t' = t, t-1, \cdots, t-W'+1$}{
				\If{$f'_{B_{\tau}^{i}}(v_{t'}) \geq \tau \wedge |B_{\tau}^{i}| < k$}{
					Add $v_{t'}$ to $B_{\tau}^{i}$\;
					$S_{\tau}^{i, t'} \leftarrow B_{\tau}^{i}$\;
				}
			}		
		}			
	}
	// Update all active $S_{\tau}^{i, t'}$\;
	$i_t \gets \max\{0, \lceil (t-W+1)/W' \rceil\}$\;
	\For{$i_t \leq i \leq \lceil t/W' \rceil -1$}{
		\For{$\tau \in T_i$}{		
			\For{$S_{\tau}^{i,t'} \in \mathcal{S}_{\tau}^i$}{
				\If{$f'_{S_{\tau}^{i, t'}}(v_t) \geq \tau \wedge |S_{\tau}^{i, t'}| < k$}{
					Add $v_t$ to $S_{\tau}^{i, t'}$
				}
			}
		}
	}
	Let $S_{\tau}^{i_t,t^*}$ be the set in $\mathcal{S}_{\tau}^{i_t}$ with minimum $t^*$ such that $S_{\tau}^{i_t, t^*} \subseteq A_t$\;
	Return $\max_{\tau \in T_{i_t}} f(S_{\tau}^{i_t, t^*})$\;
}
\end{algorithm}


\begin{theorem}\label{thm:OneHalf}
For any $\epsilon > 0$, Algorithm~\ref{alg:BidirectionalAlgorithm} with 
$\delta = \epsilon$ is a $(\nicefrac12-\epsilon)$-approximation for 
submodular maximization with cardinality constraint in sliding window model. 
\end{theorem}

\begin{proof}
The main idea is similar to the proof of Lemma~\ref{lem:ThresholdGuarantee}. 
There exists some $(1-\delta) f_k(A_t)/2k \leq \tau \leq f_k(A_t)/2k$ in set $T_{i_t}$. 
We lower bound value of $S_{\tau}^{i_t, t^*}$. If $|S_{\tau}^{i_t,t^*}| = k$, we have $f(S_{\tau}^{i_t,t^*}) \geq k\tau \geq (1-\delta)f_k(A_t)/2$ which proves the claim. 

In the other case, we show $f'_{S_{\tau}^{i_t,t^*}}(x) < \tau$ for any $x \in A_t$. 
The choice of $t^*$ implies that $f'_{S_{\tau}^{i_t,t^*}}(x) < \tau$ for any $x \in A_t$ that arrives in sub-window $i_t$ otherwise we could find a smaller $t^*$ which is a contradiction. Furthermore any $x$ that comes after sub-window $i$ with incremental value $\geq \tau$  is also added to $S_{\tau}^{i_t,t^*}$. Therefore the incremental value of any $x \in A_t$ is less than $\tau$.
Let $OPT$ be the $\argmax_{S \subset A_t: |S| \leq k} f(S)$.  
Submodularity guarantees that $\quad f(OPT) - f(S_{\tau}^{i_t,t^*}) \leq$ \\$ \sum_{x \in OPT} f'_{S_{\tau}^{i_t,t^*}}(x) < |OPT| \tau \leq k\tau \leq f(OPT)/2$ which completes the proof.
\end{proof}

We now state a bound on the memory and the update time of Algorithm ~\ref{alg:BidirectionalAlgorithm}. 

\begin{theorem}
Algorithm ~\ref{alg:BidirectionalAlgorithm} with 
$\delta = \epsilon$ has an average update time of $O(Lk\log(k\Spread)W/(W'\epsilon))$ per element and uses memory $O(W' + (k^2\log(k\Spread)W/(W'\epsilon)))$ where $L$ is an upper bound on each evaluation of function $f$. 
\end{theorem}
\ifjournal
\begin{proof}
We start by bounding $|T_i|$ for any sub-window $i$. Note that $m_i$ is at most $\log_{1+\delta}(2k\Delta/f(v_t)) = O(\log(k\Spread)/\epsilon)$. We note that Algorithm~\ref{alg:BidirectionalAlgorithm} keeps maintaining  set $S_{\tau}^{i,t'}$  only if $i$ is in range $[\lceil (t-W+1)/W' \rceil, \lceil t/W' \rceil -1]$. Therefore there are at most $W/W'$ active values of $i$, and the sets associated with smaller (older) values of $i$  can be discarded. Each of these sets has size at most $k$, and there are also $k+1$ values of $t'$. We conclude that the total memory needed to keep these sets is $O(k^2\log(k\Spread)W/(W'\epsilon))$.
In order to initialize the set $S_{\tau}^{i,t'}$, we need to have access to all elements of sub-window $i$ when we reach time $t = iW'$. This means we need to also keep the last $W'$ elements. So the total memory is $O(W' + k^2\log(k\Spread)W/(W'\epsilon))$. 
The computation time per element to keep the sets updated is upper bounded by the number of relevant sets $O(k\log(k\Spread)W/(W'\epsilon))$ times $L$. The computation time to initialize the sets is $W'|T_i|L$ which is done once for $W'$ elements. So the average computation time per element in this part is $L|T_i| = O(L\log(k\Spread)/\epsilon)$. This gives  a total average computation time per element of $O(Lk\log(k\Spread)W/(W'\epsilon))$.
\end{proof}
\fi

Note that the two last theorems imply, for example, that it is possible to obtain a $\nicefrac12 -\epsilon$ approximation using only $O(Lk\log(k\Spread)\sqrt{W}/(\epsilon))$ update time and $O(k^2\log(k\Spread)\sqrt{W}/\epsilon)$ memory.

\section{Experiments}
\label{sect:exp}

We present the experimental evaluation of our methods on several publicly available real-world datasets. We first show how to avoid some of the assumptions we made during the analysis, for example knowing the maximum marginal gain, $\Delta$. Then we describe the datasets and baselines, and finally present the empirical results. Overall, we show that our algorithms are significantly faster that the offline greedy algorithm that recomputes the results at every time step, while achieving comparable accuracy. 

\subsection{Implementation}
\label{sect:implementation}
The algorithms were implemented in \texttt{C++} and run on commodity hardware. Each run employed a single core. 

One latent assumption we made in the analysis of the algorithm is the knowledge of $\Delta$. Although the value of $\Delta$ is sometimes known, we show how to implement the algorithms without this apriori knowledge using lazy initialization. A similar approach has been used in ~\cite{KDD14}. 

We discuss the details for \StreamAlg, but note that the same method works for all the other algorithms.  The parameter $\Delta$ is only used to define the number of thresholds $T$ we use. Specifically, we set $m = \lfloor \log_{1+\delta} 2k\Delta/f(u_1) \rfloor$, and define thresholds from $\frac{f(u_1)}{2k}$ to $\frac{(1+\delta)^mf(u_1)}{2k}$. 

We can achieve the same provable guarantees while actually initializing the thresholds lazily. Let $\Delta_t = \max_{i\le t} f(\{v_i\}) $ be the maximum of the value of $f$ on any single element seen up to time $t$. Let $m_t =   \lfloor \log_{1+\delta} 2k\Delta_t/f(u_1) \rfloor$, the algorithm will maintain all thresholds in $T_t  =$ $\{\frac{f(u_1)}{2k}$, $\frac{(1+\delta)f(u_1)}{2k}$, $\frac{(1 + \delta)^2f(u_1)}{2k}$, $\ldots, \frac{(1+\delta)^{m_t}f(u_1)}{2k}\},$ and the associate solution $S_{\tau}$ for $\tau \in T_t$. Note that $m_t$ can only increase. In these cases, i.e. when $m_t > m_{t-1}$, we first add the new thresholds and initialize them to $\emptyset$. Submodularity of $f$ guarantees that no prior elements would meet the new thresholds. 


The same technique can be used for \SmoothAlg, where again each individual copy of the algorithm \StreamAlg lazily initializes the thresholds depending on the running maximum. Similarly, in \BidirectionalAlg we maintain the $T_i$ and the associated sets $S^{i,t}_\tau$ by initializing them only when necessary. Finally, observe that for specific cases of $f$, the algorithms can be further sped up by discarding small thresholds. For example, in unweighted coverage, the minimum non-zero $f'(\cdot)$ is at least $1$, therefore all smaller thresholds can be ignored.  

\ifjournal 
\begin{figure*}[ht!]
\begin{center}
\subfigure[DBLP]{\includegraphics[width=0.4\textwidth,keepaspectratio]{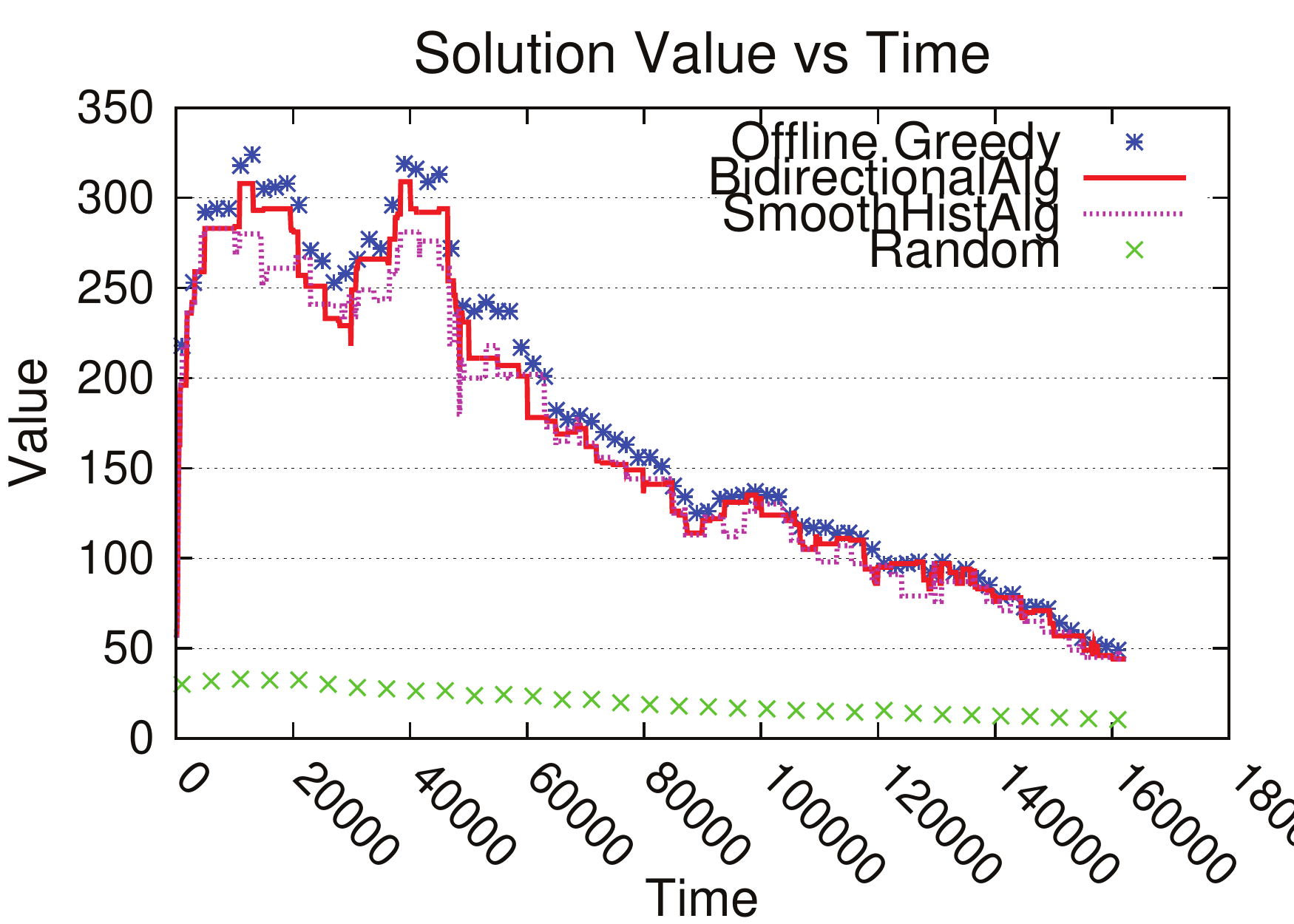}}
\subfigure[Gowalla{\bf}]{\includegraphics[width=0.4\textwidth,keepaspectratio]{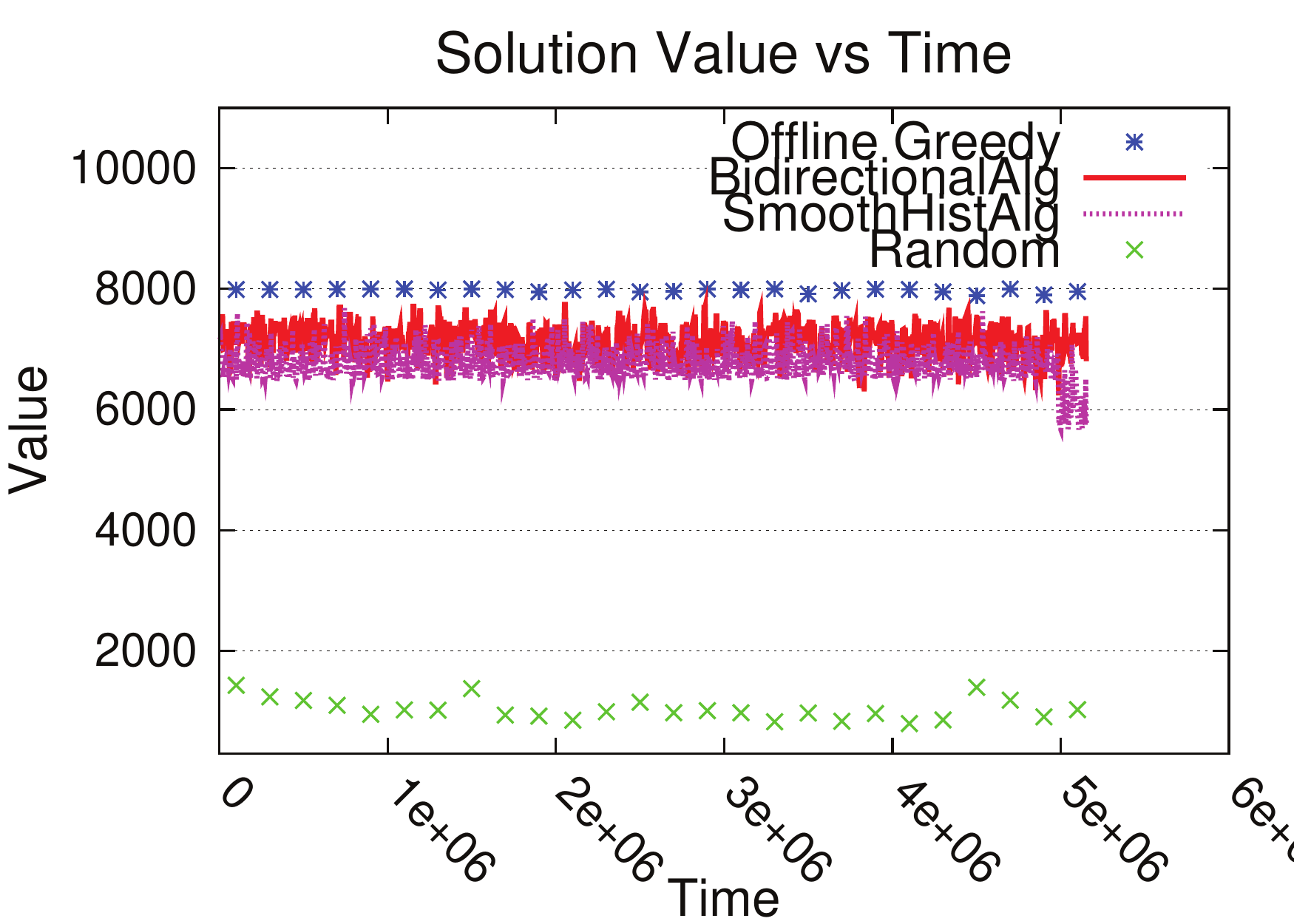}}
\subfigure[Yahoo!{\bf}]{\includegraphics[width=0.4\textwidth,keepaspectratio]{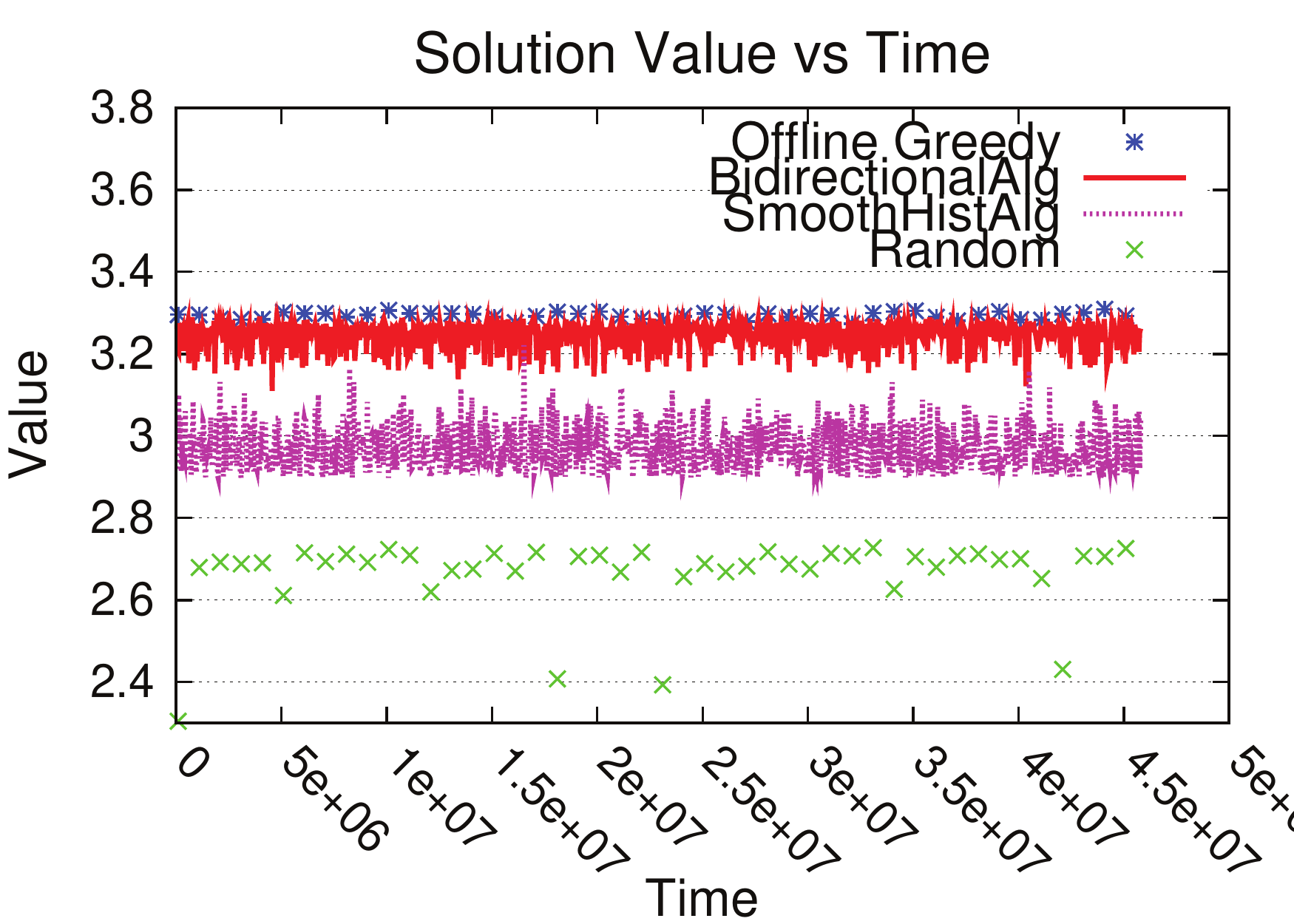}}
\caption{Value of the solution obtained by our algorithms for $k=10$, $W=10{,}000$ and $\epsilon=0.1$ as well as by the off-line greedy algorithm and a random baseline on a sample of time-steps.  Notice how our algorithms achieve solutions with value close to that of off-line greedy.}
\label{fig:evolution}
\end{center}
\end{figure*}
\fi 

\subsection{Datasets}

In this section we describe the datasets used for the evaluation. 

{\bf DBLP. } We define a max cover instance from the  DBLP publication records~\cite{DBLP}. We extract $1.8$ million publications and $1.2$ million authors for the period from $1959$ to $2016$. Our goal is to maintain a set of $k$ authors  that together represent the largest possible number of different conferences in computer science. We say that a conference is represented by an author if she has published at least 3 papers in the venue, this gives about $160$ thousand items in the stream (one per author with at least $3$ papers in the same conference).  We order the authors by the time of their first publication (with ties broken randomly). 

{\bf Gowalla.} In this experiment we want to simulate a system that maintains a set of currently {\it hot} locations that cover as many users as possible in their immediate proximity. We use the check in data collected by the Gowalla social network~\cite{cho2011friendship}, which contains about $6.5$ millions timestamped and geo-localized check-ins of about $200$ thousand users over the period of 2/2009 - 10/2010. We first partition the dataset into two parts temporally. We use the first 20\% to define the submodular function $f$, as we describe below, and the last 80\% to evaluate our algorithms. 

During the first part, we divide the globe latitude and longitude coordinates into a uniform degree-spaced grid of size $80{,}000 \times 80{,}000$ cells (these correspond to $1$km size cells at the Equator). For each cell $(i,j)$ in the grid we record the set of users that had at least one check-in in that cell. For a given location $(i, j)$, the associated set is the set of users that checked in to a place in location $(i,j)$ or an immediately adjacent cell during the first phase. The goal is to maintain $k$ check-in locations from the active window that covers as many distinct users as possible.\footnote{We recognize that our modeling of this problem is simplistic, but wanted to keep it to a minimum, as the main purpose of this dataset is to evaluate the performance of the algorithms, and not to  study location-based systems.} 

{\bf Yahoo! Front Page Visits.} Here we experiment with a standard dataset used in submodular optimization literature~\cite{KDD14}. This dataset is extracted from the click logs of news  articles displayed in the Yahoo! Front Page~\cite{yahoo-front}. It contains $46$ millions timestamped $5$-dimensional feature vectors (we discard the constant feature and normalize the vector norm), representing user-visits over the period of ten days in May 2009. We stream the vectors in time order and optimize the active set selection function defined in section~\ref{sect:applications}. 

\subsection{Results}
 \label{sect:results} 
To evaluate the performance of our algorithms, we consider two benchmarks. The first, serving as a sanity check, is a random sample of $k$ points from the sliding window. The second is the batch greedy algorithm on the elements in the active set. The latter serves as an upper bound, as it is the best algorithm for the problem. However, since it is not optimized for streaming computations, it is expensive to evaluate. As such, we run it regularly, but not at every time step. We emphasize that ours are the first algorithms that handle streams with both additions and deletions. 

{\bf Value of the output over time.}
In our first experiment we show the value of the objective function at every time step as computed by the algorithm and the two benchmarks. For the random baseline, we average the results over $1000$ trials, all of the other algorithms are deterministic. We set $W = 10{,}000$, $k = 10$, and $\epsilon = 0.1$.  The results are shown in Figure~\ref{fig:evolution}. Notice that in all the experiments involving \BidirectionalAlg  we set $W'=W$ to model the scenario of a user that wants the best running time for a $1/2-\epsilon$ approximation. 

In all instances the \BidirectionalAlg algorithm results are very close to the off-line greedy algorithm. As expected, the solution of \SmoothAlg is slightly worse (we observe a gap of about 10\%). So both algorithms perform much better than the pessimistic worst-case analysis, a result that is quantitatively confirmed in the next section. Not surprisingly, all algorithms greatly exceed the random baseline. 

Finally note that for the DBLP dataset, the solution value generally decreases, as authors who first publish later tend to have shorter careers, and thus have not had a chance to cover as many venues. On the other hand, due to the nature of the objective, the value of the solution in Gowalla and Yahoo! datasets remains relatively stable, and oscillates in a smaller region. 

\ifjournal 
\begin{figure}[ht!]
\begin{center}
\subfigure[\BidirectionalAlg]{\includegraphics[width=0.4\textwidth,keepaspectratio]{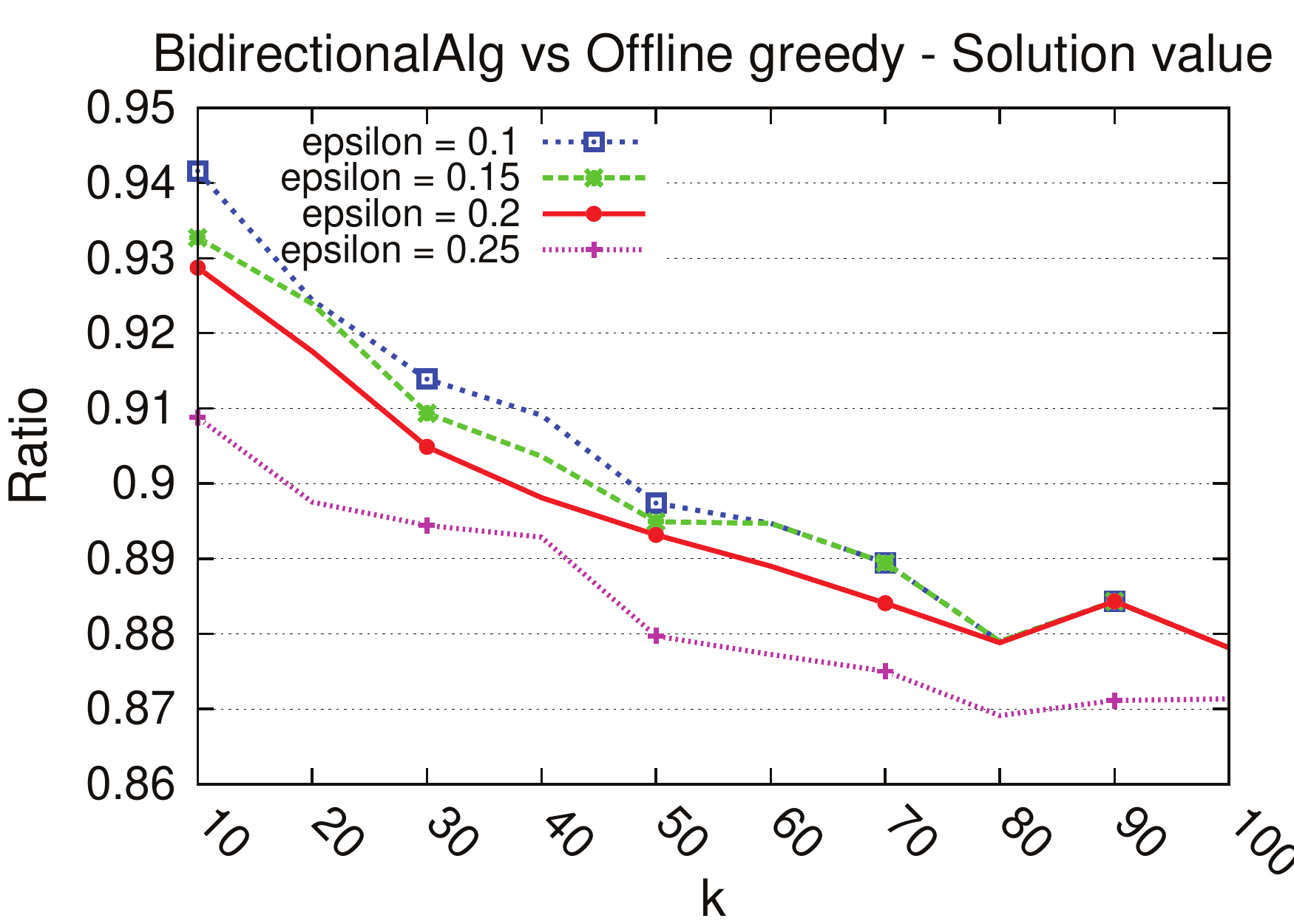}}
\subfigure[\SmoothAlg]{\includegraphics[width=0.4\textwidth,keepaspectratio]{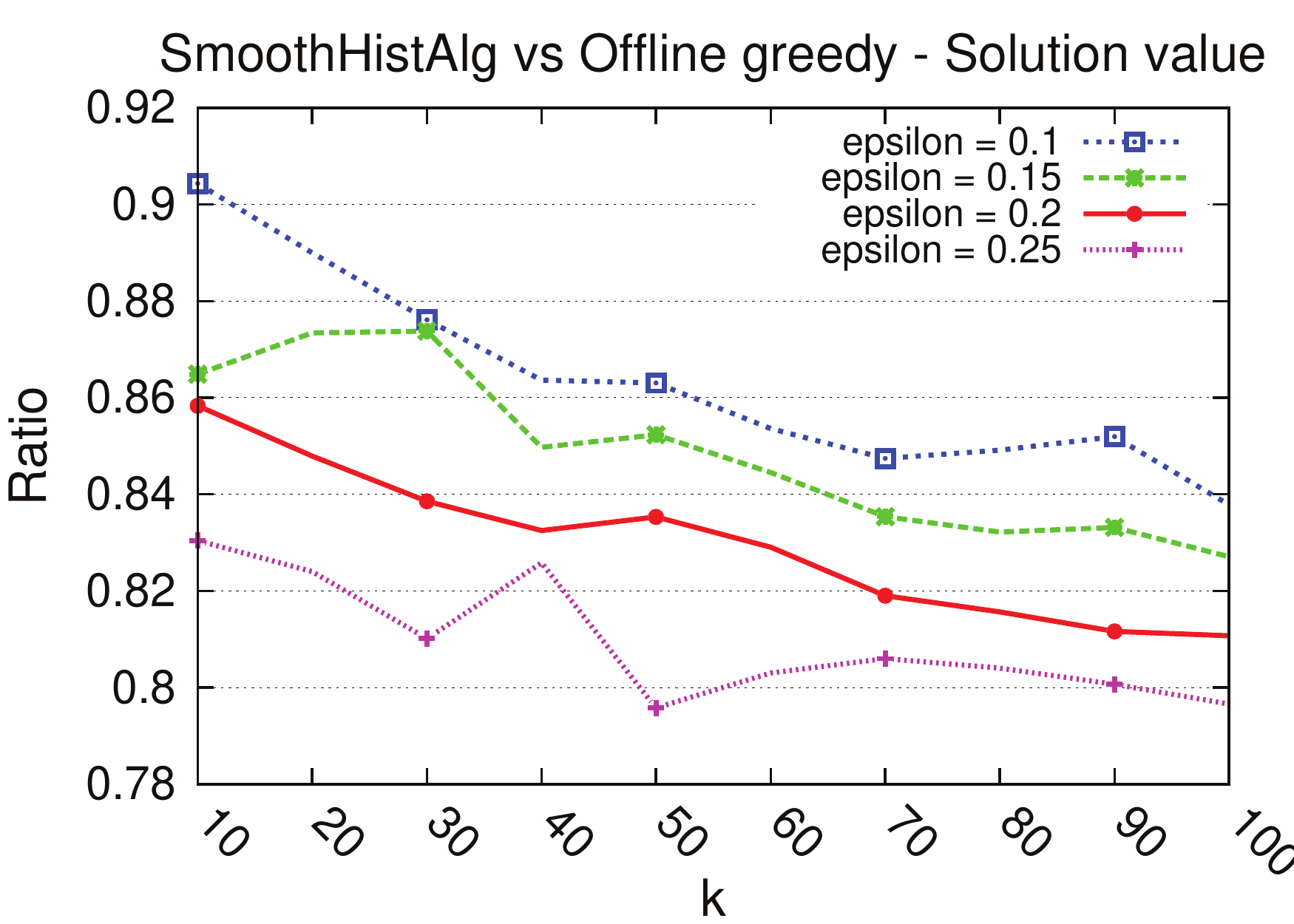}}
\caption{Average ratio of solution value of our algorithms over offline greedy in DBLP---higher is better}
\label{fig:value-cmp}
\end{center}
\end{figure}
\fi

{\bf Comparison with greedy.} To better understand the relative performance of the algorithms, we focus on the DBLP dataset, and consider what fraction of the benchmark greedy solution is achieved by all algorithms for different values of $k$; we plot the results in Figure~\ref{fig:value-cmp}. Our algorithms always report solutions that are between $80\%$ and $95\%$ of the value of offline greedy for any setting of $\epsilon \in [0.1, 0.25]$ and any $k \in [10,100]$, far exceeding the theoretical worst case analysis. All of the results match the intuition provided by the theory: for the same $\epsilon$ parameter the \BidirectionalAlg returns higher values than \SmoothAlg, and lower $\epsilon$ parameters yield better solutions. Also while the problem becomes more challenging with higher $k$ values (more overlap in the sets need to be handled) the streaming algorithms achieve good results (similar results are observed in all datasets and using worst-case ratios instead of average ratios).

Finally, we evaluate the speed of our algorithms, again as compared to the offline greedy approach. Following previous work~\cite{KDD14}, we record the average number of evaluations of the submodular function executed for each item processed. This captures the most expensive operations, and ignores implementation variations. We show the results in Figure~\ref{fig:calls-cmp}.

Notice that our algorithms are much faster than re-running greedy from scratch. Even in small datasets with small window size,  our algorithms require between a factor of $2{,}000$ and $6{,}000$ fewer calls per item processed. Even larger speed-ups can be observed for larger datasets and window sizes. As expected,  the speedups increase with $\epsilon$.  We observe that as $k$ increases, the speedups achieved by the \SmoothAlg algorithm grow as well, while those of \BidirectionalAlg are slightly decreasing with $k$. This is expected as \SmoothAlg update time depends only poly-logarithmically on $k$ while \BidirectionalAlg has a linear dependence. These considerations are confirmed by the results in Figure~\ref{fig:calls-number} where we report the average number of evaluation of the submodular function in the setting of the previous experiment. Notice how only a few hundred evaluations are sufficient, and that the evaluations executed by \SmoothAlg are less sensitive to $k$ as expected.

\ifjournal
\begin{figure}[t!]
\begin{center}
\subfigure[\BidirectionalAlg]{\includegraphics[width=0.4\textwidth,keepaspectratio]{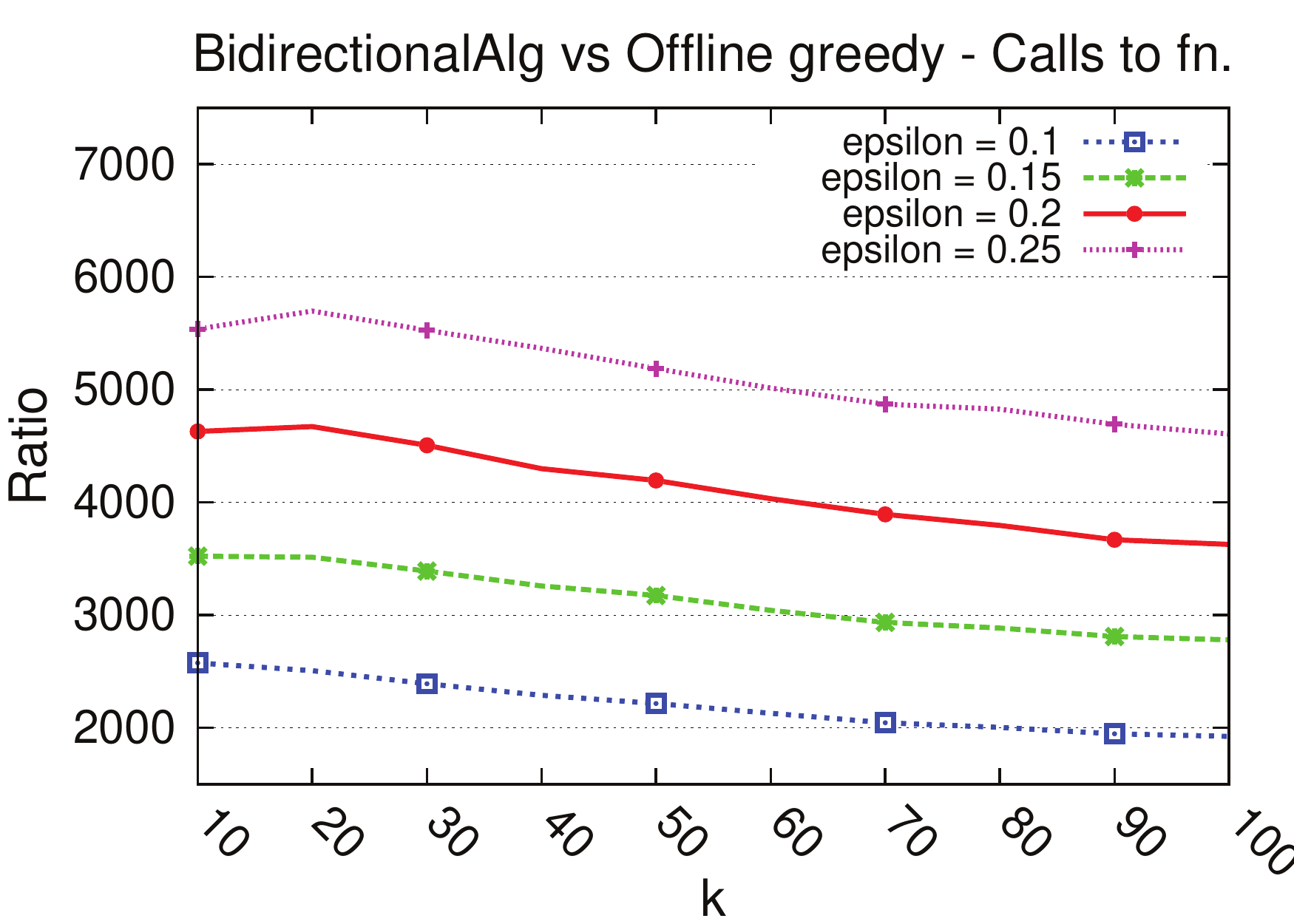}}
\subfigure[\SmoothAlg]{\includegraphics[width=0.4\textwidth,keepaspectratio]{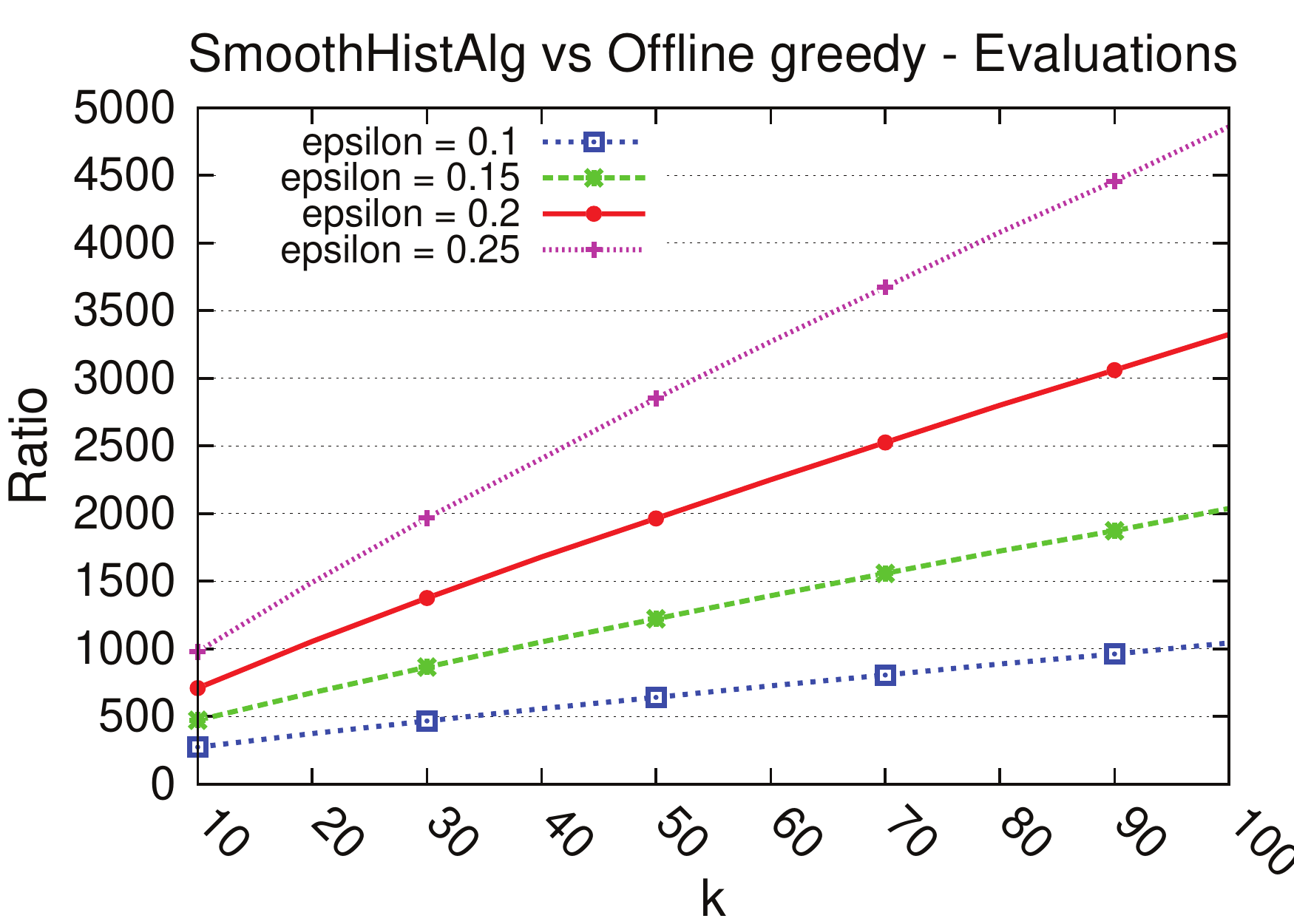}}
\caption{{Average ratio of the number of submodular function evaluations executed by greedy over the  ones executed by our algorithms---higher is better.}}
\label{fig:calls-cmp}
\end{center}
\end{figure}
\begin{figure}[t!]
\begin{center}
\subfigure[\BidirectionalAlg]{\includegraphics[width=0.4\textwidth,keepaspectratio]{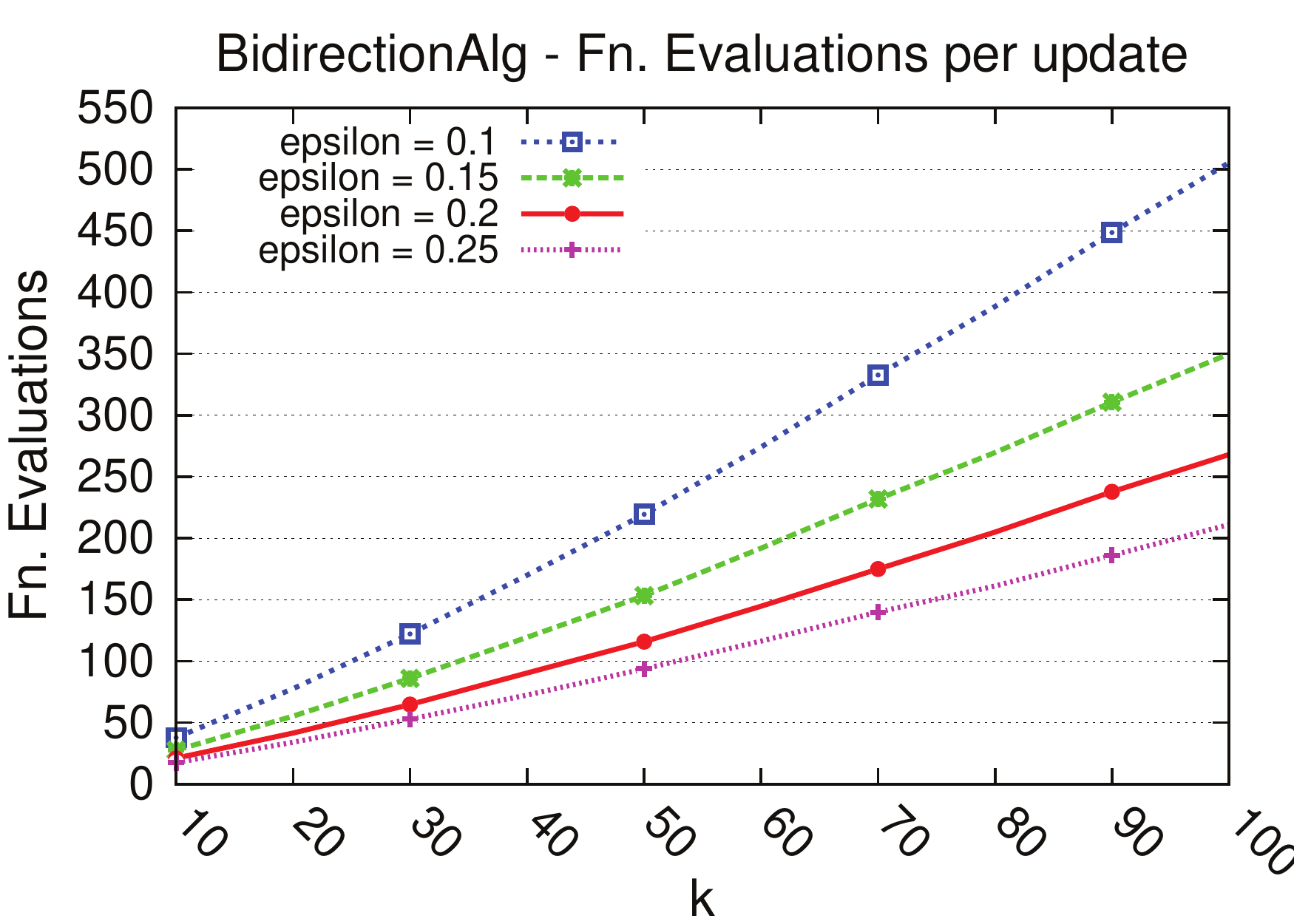}}
\subfigure[\SmoothAlg]{\includegraphics[width=0.4\textwidth,keepaspectratio]{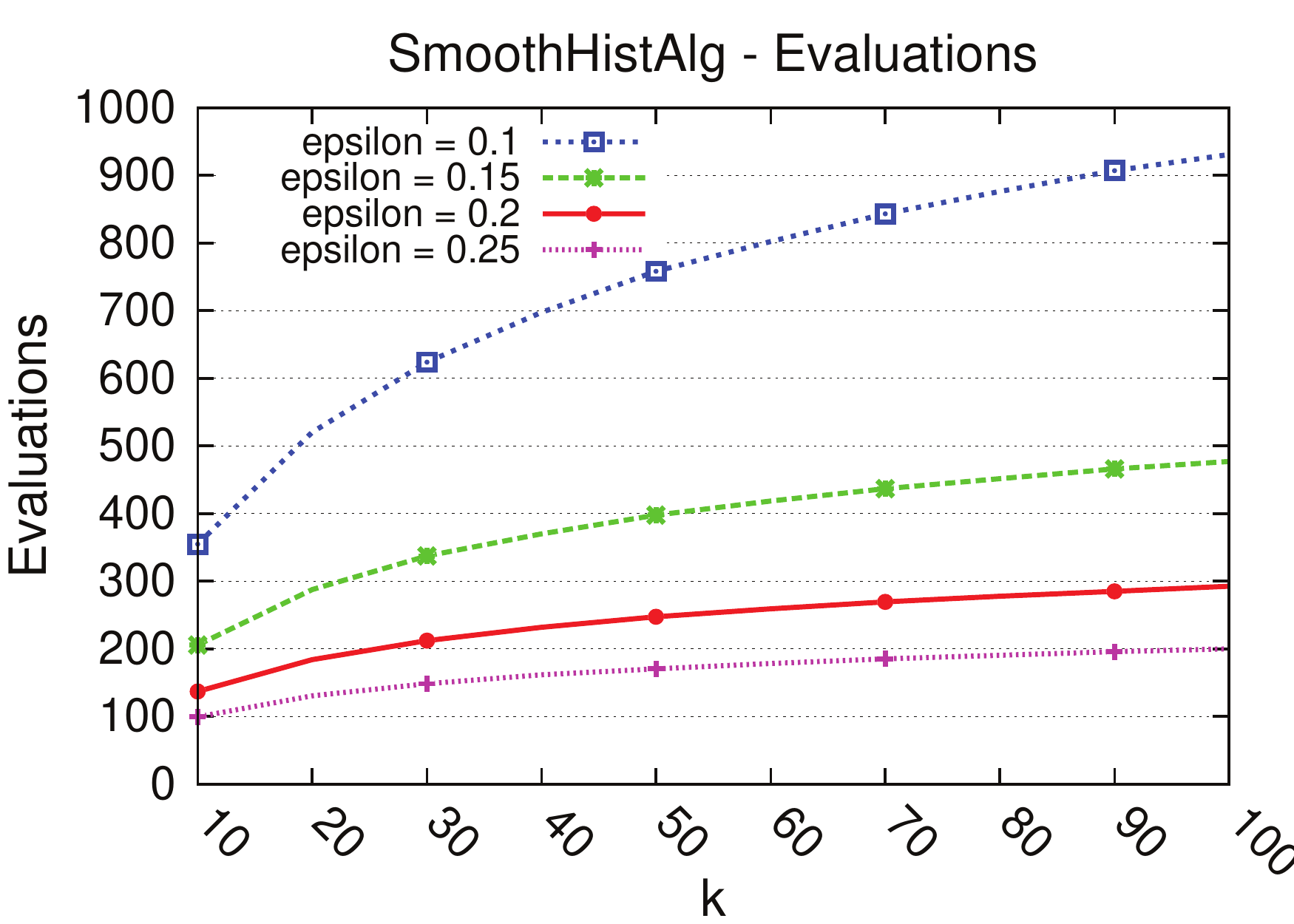}}
\caption{{Number of evaluations of the submodular function executed per update---lower is better.}}
\label{fig:calls-number}
\end{center}
\end{figure}

\fi

{\bf Scalability.}
Having established that our algorithms preserve solutions with quality close to the offline greedy algorithm, while taking significantly less time, we  evaluate the scalability of \SmoothAlg in terms of memory use and speed on larger datasets with more challenging window sizes. We set $W=1{,}000{,}000$ and run \SmoothAlg on Gowalla and Yahoo using $\epsilon = 0.25$.
First, in Figure~\ref{fig:others-calls} we show the average number of function evaluation per item processed as a function of $k$. The conclusions from previous experiment continue to hold, with our algorithm requiring no more than $200$ function calls to process every item. Then, we evaluate the memory requirement of our sublinear algorithm \SmoothAlg. To do so in a platform independent way we compute the total number of items stored by our algorithm in all sets $S_\tau$ (counting repetitions) at time $t$, and look at the maximum over the entire stream. We report the results in Figure~\ref{fig:others-stored}. Observe that our algorithm maintains only a small fraction of the current sliding window (between $0.05\%$ and $0.4\%$) thus allowing to process large sliding windows with minimal memory. 

\ifjournal
\begin{figure}[t!]
\begin{center}
\subfigure[Function Evaluations per Update\label{fig:others-calls}]{\includegraphics[width=0.4\textwidth,keepaspectratio]{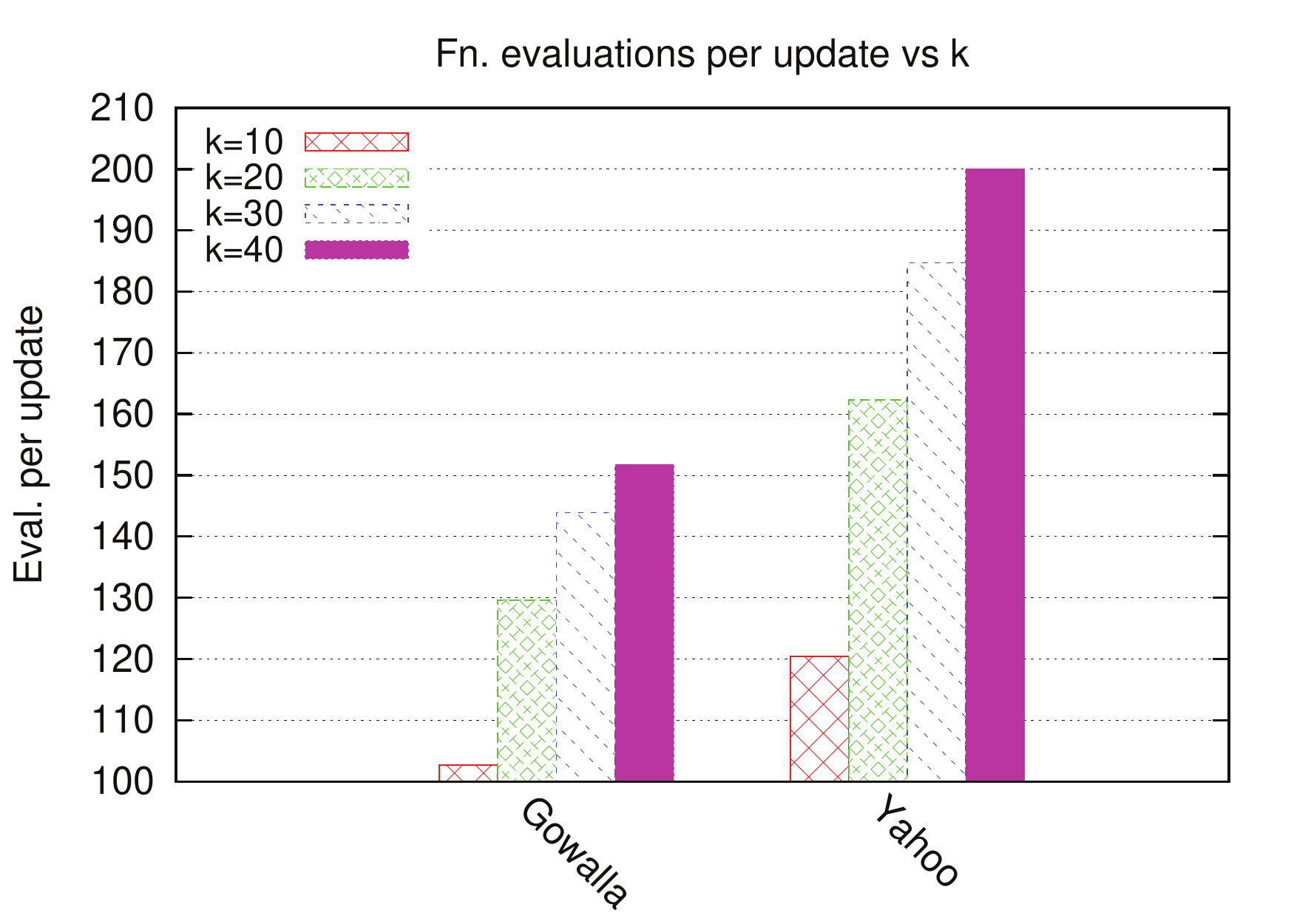}}
\subfigure[Max. Fraction of the Window Stored\label{fig:others-stored}]{\includegraphics[width=0.4\textwidth,keepaspectratio]{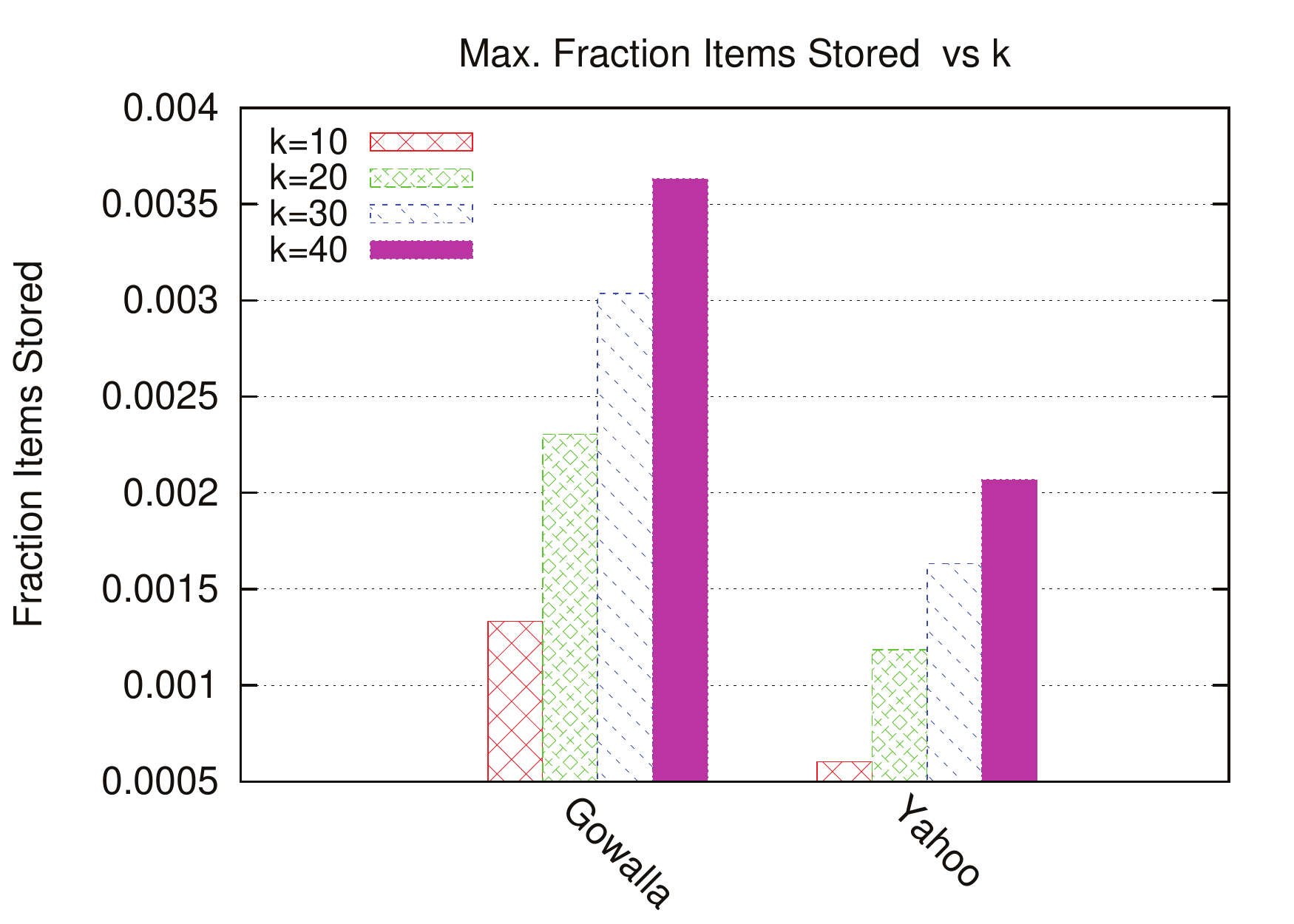}}
\caption{{Number of evaluations of the submodular function per update and fraction of items in the window stored by \SmoothAlg---lower is better.}}
\label{fig:others}
\end{center}
\end{figure}

\fi

\section{Conclusions}
\label{sect:conclusions}
We showed the first non-trivial algorithms for arbitrary monotone submodular functions subject to cardinality constraints in sliding window settings. We proved that one can achieve approximation ratios of $\nicefrac{1}{2}-\epsilon$, while using sublinear space and time per update.  An interesting direction for future work is to address this problem in the fully dynamic setting, where addition and deletion of items is allowed in arbitrary order. Another interesting question is whether it is possible to improve the approximation guarantees of $\nicefrac1 2-\epsilon$ in the streaming context.
\balance
\bibliography{references}

\begin{thebibliography}{10}

\bibitem{sublinear}
{List of open problems in sublinear algorithms: Problem 20}.
\newblock \url{http://sublinear.info/20}.

\bibitem{mirrokniKDD2013}
Zeinab Abbassi, Vahab~S Mirrokni, and Mayur Thakur.
\newblock Diversity maximization under matroid constraints.
\newblock In {\em ACM SIGKDD}, pages 32--40. ACM, 2013.

\bibitem{aggarwal2007data}
Charu~C Aggarwal.
\newblock {\em Data streams: models and algorithms}, volume~31.
\newblock Springer Science \& Business Media, 2007.

\bibitem{KDD14}
Ashwinkumar Badanidiyuru, Baharan Mirzasoleiman, Amin Karbasi, and Andreas
  Krause.
\newblock Streaming submodular maximization: Massive data summarization on the
  fly.
\newblock In {\em ACM SIGKDD}, pages 671--680. ACM, 2014.

\bibitem{badanidiyuru2014fast}
Ashwinkumar Badanidiyuru and Jan Vondr{\'a}k.
\newblock Fast algorithms for maximizing submodular functions.
\newblock In {\em SODA}, pages 1497--1514. Society for Industrial and Applied
  Mathematics, 2014.

\bibitem{bakshy2012role}
Eytan Bakshy, Itamar Rosenn, Cameron Marlow, and Lada Adamic.
\newblock The role of social networks in information diffusion.
\newblock In {\em International conference on World Wide Web}, pages 519--528.
  ACM, 2012.

\bibitem{bateni2010submodular}
MohammadHossein Bateni, MohammadTaghi Hajiaghayi, and Morteza Zadimoghaddam.
\newblock Submodular secretary problem and extensions.
\newblock In {\em Approx}, pages 39--52. Springer, 2010.

\bibitem{DBLP:conf/soda/BravermanLLM16}
Vladimir Braverman, Harry Lang, Keith Levin, and Morteza Monemizadeh.
\newblock Clustering problems on sliding windows.
\newblock In {\em SODA}, pages 1374--1390, 2016.

\bibitem{Braverman07}
Vladimir Braverman and Rafail Ostrovsky.
\newblock Smooth histograms for sliding windows.
\newblock In {\em FOCS}, pages 283--293, 2007.

\bibitem{DBLP:journals/algorithmica/ChanLLT12}
Ho{-}Leung Chan, Tak~Wah Lam, Lap{-}Kei Lee, and Hing{-}Fung Ting.
\newblock Continuous monitoring of distributed data streams over a time-based
  sliding window.
\newblock {\em Algorithmica}, 62(3-4):1088--1111, 2012.

\bibitem{cho2011friendship}
Eunjoon Cho, Seth~A Myers, and Jure Leskovec.
\newblock Friendship and mobility: user movement in location-based social
  networks.
\newblock In {\em ACM SIGKDD}, pages 1082--1090. ACM, 2011.

\bibitem{datar2002maintaining}
Mayur Datar, Aristides Gionis, Piotr Indyk, and Rajeev Motwani.
\newblock Maintaining stream statistics over sliding windows.
\newblock {\em SIAM journal on computing}, 31(6):1794--1813, 2002.

\bibitem{DBLP}
DBLP.
\newblock {DBLP xml dump -- Sept 2016}.
\newblock \url{http://dblp.uni-trier.de/xml/}.

\bibitem{Feige98}
Uriel Feige.
\newblock A threshold of ln n for approximating set cover.
\newblock {\em J. ACM}, 45(4):634--652, July 1998.

\bibitem{FreyNIPS2005}
Brendan~J. Frey and Delbert Dueck.
\newblock Mixture modeling by affinity propagation.
\newblock In {\em NIPS}, pages 379--386, Cambridge, MA, USA, 2005. MIT Press.

\bibitem{gibbons2002distributed}
Phillip~B Gibbons and Srikanta Tirthapura.
\newblock Distributed streams algorithms for sliding windows.
\newblock In {\em SPAA}, pages 63--72. ACM, 2002.

\bibitem{KempeKDD2003}
David Kempe, Jon Kleinberg, and \'{E}va Tardos.
\newblock Maximizing the spread of influence through a social network.
\newblock In {\em ACM SIGKDD}, pages 137--146, New York, NY, USA, 2003. ACM.

\bibitem{kumar2015fast}
Ravi Kumar, Benjamin Moseley, Sergei Vassilvitskii, and Andrea Vattani.
\newblock Fast greedy algorithms in mapreduce and streaming.
\newblock {\em ACM Transactions on Parallel Computing}, 2(3):14, 2015.

\bibitem{leskovec2007cost}
Jure Leskovec, Andreas Krause, Carlos Guestrin, Christos Faloutsos, Jeanne
  VanBriesen, and Natalie Glance.
\newblock Cost-effective outbreak detection in networks.
\newblock In {\em ACM SIGKDD}, pages 420--429. ACM, 2007.

\bibitem{BilmesACL2011}
Hui Lin and Jeff Bilmes.
\newblock A class of submodular functions for document summarization.
\newblock In {\em ACL: Human Language Technologies}, pages 510--520.
  Association for Computational Linguistics, 2011.

\bibitem{mcgregor2014graph}
Andrew McGregor.
\newblock Graph stream algorithms: a survey.
\newblock {\em ACM SIGMOD}, 43(1):9--20, 2014.

\bibitem{citeulike:10637610}
Michel Minoux.
\newblock {Accelerated greedy algorithms for maximizing submodular set
  functions}.
\newblock In J.~Stoer, editor, {\em Optimization Techniques}, volume~7 of {\em
  Lecture Notes in Control and Information Sciences}, chapter~27, pages
  234--243. Springer Berlin Heidelberg, Berlin/Heidelberg, 1978.

\bibitem{KrauseNIPS2013}
Baharan Mirzasoleiman, Amin Karbasi, Rik Sarkar, and Andreas Krause.
\newblock Distributed submodular maximization: Identifying representative
  elements in massive data.
\newblock In {\em NIPS}, pages 2049--2057. Curran Associates, Inc., 2013.

\bibitem{nemhauser1978analysis}
George~L Nemhauser, Laurence~A Wolsey, and Marshall~L Fisher.
\newblock An analysis of approximations for maximizing submodular set
  functions—i.
\newblock {\em Mathematical Programming}, 14(1):265--294, 1978.

\bibitem{saha2009maximum}
Barna Saha and Lise Getoor.
\newblock On maximum coverage in the streaming model \& application to
  multi-topic blog-watch.
\newblock In {\em SDM}, volume~9, pages 697--708. SIAM, 2009.

\bibitem{DBLP:journals/algorithms/TingLCL11}
Hing{-}Fung Ting, Lap{-}Kei Lee, Ho{-}Leung Chan, and Tak~Wah Lam.
\newblock Approximating frequent items in asynchronous data stream over a
  sliding window.
\newblock {\em Algorithms}, 4(3):200--222, 2011.

\bibitem{yahoo-front}
Yahoo!
\newblock {Front Page Today Module User Click Log Dataset}.
\newblock \url{http://webscope.sandbox.yahoo.com}.

\end{thebibliography}
\bibliographystyle{plain}

\end{document}